\documentclass[journal]{IEEEtran}
\usepackage{graphicx}
\usepackage{multirow}
\usepackage{bm}
\usepackage{subfigure}
\usepackage{amssymb}
\usepackage{algorithm}
\usepackage{algorithmicx}
\usepackage{algpseudocode}
\usepackage{amsthm}
\usepackage{setspace}
\usepackage{balance}

\newtheorem{theorem}{Theorem}

\newtheorem{lemma}{Lemma}

\usepackage{cite}

\ifCLASSINFOpdf
\else
\fi
\usepackage[cmex10]{amsmath}
\begin{document}
%
\title{Sparse Multipath Channel Estimation and Decoding for Broadband Vector OFDM Systems}
%
%
%
\author{Qi~Feng,
        Xiang-Gen~Xia,~\IEEEmembership{Fellow,~IEEE,}
        Zhihui~Ye,
        and~Naitong~Zhang%
\thanks{This work was supported by the China Scholarship Council (Grant No. 201406190083), in part by the Satellite Communication and Navigation Collaborative Innovation Center (Grant No. SatCN-201408), in part by the Program B for Outstanding PhD Candidate of Nanjing University (Grant No. 201501B013), and in part by the Special Research Foundation of Marine Public Service Sector (Grant No. 201205035).}%
\thanks{Q. Feng and N. Zhang are with the School of Electronic Science and Engineering, Nanjing University, Nanjing 210023,
P. R. China (e-mail: xiaoyehouzi@126.com, ntzhang@nju.edu.cn).}%
\thanks{X.-G. Xia is with the Department of Electrical and Computer Engineering, University of Delaware, Newark, DE 19716, USA (e-mail: xxia@ee.udel.edu).}%
\thanks{Z. Ye is with the School of Electronic and Optical Engineering, Nanjing University of Science and Technology, Nanjing 210094, P. R. China. She is also with the School of Electronic Science and Engineering, Nanjing University, Nanjing 210023,
P. R. China (e-mail: zhye@nju.edu.cn).}}

\maketitle

\begin{abstract}
Vector orthogonal frequency division multiplexing (V-OFDM) is a general system that builds a bridge between OFDM and single-carrier frequency domain equalization in terms of intersymbol interference and receiver complexity. In this paper, we investigate the sparse multipath channel estimation and decoding for broadband V-OFDM systems. Unlike the non-sparse channel estimation, sparse channel estimation only needs to recover the nonzero taps with reduced complexity. We first consider a simple noiseless case that the pilot signals are transmitted through a sparse channel with only a few nonzero taps, and then consider a more practical scenario that the pilot signals are transmitted through a sparse channel with additive white Gaussian noise interference. The exactly and approximately sparse inverse fast Fourier transform (SIFFT) can be employed for these two cases. The SIFFT-based algorithm recovers the nonzero channel coefficients and their corresponding coordinates directly, which is significant to the proposed partial intersection sphere (PIS) decoding approach. Unlike the maximum likelihood (ML) decoding that enumerates symbol constellation and estimates the transmitted symbols with the minimum distance, the PIS decoding first generates the set of possible transmitted symbols and then chooses the transmitted symbols only from this set with the minimum distance. The diversity order of the PIS decoding is determined by not only the number of nonzero taps, but also the coordinates of nonzero taps, and the bit error rate (BER) is also influenced by vector block size to some extent but roughly independent of the maximum time delay. Simulation results indicate that by choosing appropriate sphere radius, the BER performance of the PIS decoding outperforms the conventional zero-forcing decoding and minimum mean square error decoding, and approximates to the ML decoding with the increase of signal-to-noise ratio, but reduces the computational complexity significantly.
\end{abstract}


%
\IEEEpeerreviewmaketitle

\section{Introduction}
%
%
%
%
\IEEEPARstart{O}{rthogonal} frequency division multiplexing (OFDM) has been widely adopted in both cellular systems, such as Long-Term Evolution (LTE) and Wi-Fi systems \cite{Cimini1985,Astely2009}. The main advantage of OFDM modulation is to convert an intersymbol interference (ISI) channel into multiple ISI-free subchannels and thus reduces the demodulation complexity at the receiver \cite{Hwang2009}. However, since each symbol is only transmitted over a parallel flat fading subchannel, the conventional OFDM technique may not collect multipath diversity, it thus performs worse than single carrier transmission \cite{Wang2005}. Furthermore, OFDM has high peak-to-average power ratio (PAPR) of the transmitted signals, which may affect its applications in broadband wireless communications. Single-carrier frequency domain equalization (SC-FDE) is an alternative approach to deal with ISI with low transmission PAPR \cite{Pancaldi2008}. However, induced by both fast Fourier transform (FFT) and inverse FFT (IFFT) operations at the receiver, SC-FDE suffers from the drawback that transmitter and receiver have unbalanced complexities \cite{Falconer2002}. As a result, OFDM is more suitable for downlink with high transmission speed, whereas SC-FDE can be applied for uplink that reduces PAPR and transmitter complexity as in LTE \cite{Astely2009}.

Vector OFDM (V-OFDM) for single transmit antenna systems first proposed in \cite{Xia2001} converts an ISI channel to multiple vector subchannels where the vector size is a pre-designed parameter and flexible. For each vector subchannel, the information symbols of a vector may be (are) ISI together. Since the vector size is flexible, when it is $1$, V-OFDM coincides with the conventional OFDM. When the vector size is $2$, each vector subchannel may have two information symbols in ISI. When the vector size is large enough, say, the same as the IFFT size, then the maximal number of information symbols are in ISI and it is then equivalent to SC-FDE. Therefore, V-OFDM naturally builds a bridge between OFDM and SC-FDE in terms of both ISI level and receiver complexity \cite{Xia2001,Li2012}, and it has attracted recent interests. For V-OFDM, an adaptive vector channel allocation scheme was proposed for V-OFDM systems \cite{Zhang2002}. Some key techniques, such as carrier/sampling-frequency synchronization and guard-band configuration in V-OFDM system were designed and made comparisons with the conventional OFDM systems \cite{Zhang2005}. Iterative demodulation and decoding under turbo principle is an efficient way for V-OFDM receiver \cite{Zhang2006}. Constellation-rotated V-OFDM was proposed with improved multipath diversity \cite{Han2010,Han2014,Cheng2011}. Linear receivers and the corresponding diversity order analyses are recently given in \cite{Li2012} and phase noise influence is investigate in \cite{Ngebani2014}.

For a very broadband channel, the IFFT size of an OFDM system needs to be very large, which may cause practical implementation problems, such as high PAPR and high complexity. In contrast, for a V-OFDM, its IFFT size can be fixed and independent of a bandwidth, while its vector size can be increased to accommodate the increased bandwidth. In this paper, we are interested in V-OFDM over a broadband sparse channel in the sense that it has a large time delay spread but only a few of nonzero taps \cite{Win2002,Schreiber1995,Berger2010}. For sparse channels, there have been many studies in the literature, see for example \cite{Cotter2002,Raghavendra2005,Bajwa2010,Kannu2011,Hassibi2005,Vikalo2005,Barik2014,Prasad2014}. Recently, Sparse FFT (SFFT) theory was proposed by the Computer Science \& Artificial Intelligence Lab, Massachusetts Institute of Technology \cite{Hassanieh2012Jan}. If a signal has a small number $k$ of nonzero Fourier coefficients, the output of the Fourier transform can be represented succinctly using only $k$ coefficients. For such signals, the runtime is sublinear in the signal size $n$ rather than $\mathcal O(n\log n)$. Furthermore, several new algorithms for SFFT are presented, i.e., an $\mathcal O(k\log n)$-time algorithm for the exactly $k$-sparse case, an $\mathcal O(k\log n\log(n/k))$-time algorithm for the general case \cite{Hassanieh2014May}. In this paper, a sparse channel estimation and decoding scheme for V-OFDM systems is proposed. Inspired by the idea of SFFT under the condition of signals with only a few nonzero Fourier coefficients, we first use pilot symbols to obtain channel frequency response (CFR), and then estimate channel impulse response (CIR) by using sparse IFFT (SIFFT). Based on the estimation of nonzero channel coefficients and their corresponding coordinates, an efficient partial intersection sphere (PIS) decoding is investigated and it achieves the same diversity order as the maximum likelihood (ML) decoding. The main contributions of the paper are summarized as follows.

\begin{itemize}
\item We find a connection between the exactly and approximately sparse channel models in the estimation of a V-OFDM sparse multipath channel. For a multipath channel with only a few nonzero taps, if there is no noise during the transmission, then the sparse channel can be estimated by the exactly sparse multipath channel algorithm that corresponds to Algorithm 3.1 for exactly sparse FFT in \cite{Hassanieh2014May}. When there is additive white Gaussian noise (AWGN) during the transmission, the sparse multipath channel can be estimated by the approximately sparse multipath channel algorithm that corresponds to Algorithms 4.1$-$4.2 for generally sparse FFT in \cite{Hassanieh2014May}.
\item By using the SIFFT-based algorithms, one can directly recover the nonzero channel coefficients and their corresponding coordinates, which is significant to the PIS decoding process.
\item For the PIS decoding in V-OFDM systems, the bit error rate (BER) is dependent of $K$ nonzero taps in a sparse channel and the vector size $M$ to some extent, but roughly independent of the maximum delay $D$.
\item For any given small sphere radius, the proposed PIS decoding and ML decoding are of the same diversity order, which is equal to the cardinality of the set of reminder coordinates after mod $M$, but the PIS decoding can substantially reduce the computational complexity with probability $1$.
\end{itemize}

The reminder of the paper is organized as follows. In Section II, the system model of V-OFDM is reviewed. In Section III, SIFFT-based channel estimation schemes for the exactly sparse case and the approximately sparse case are introduced. In Section IV, a PIS decoding for V-OFDM systems is proposed and analyzed. In Section V, simulation results are presented and discussed. In Section VI, this paper is concluded.

\section{System Model}
We first briefly recall a V-OFDM system for single transmit antenna, which is shown in Fig. 1. The description of system model follows the notations in \cite{Li2012} below.
\begin{figure}[t]
\centering
\includegraphics[width=3.5in]{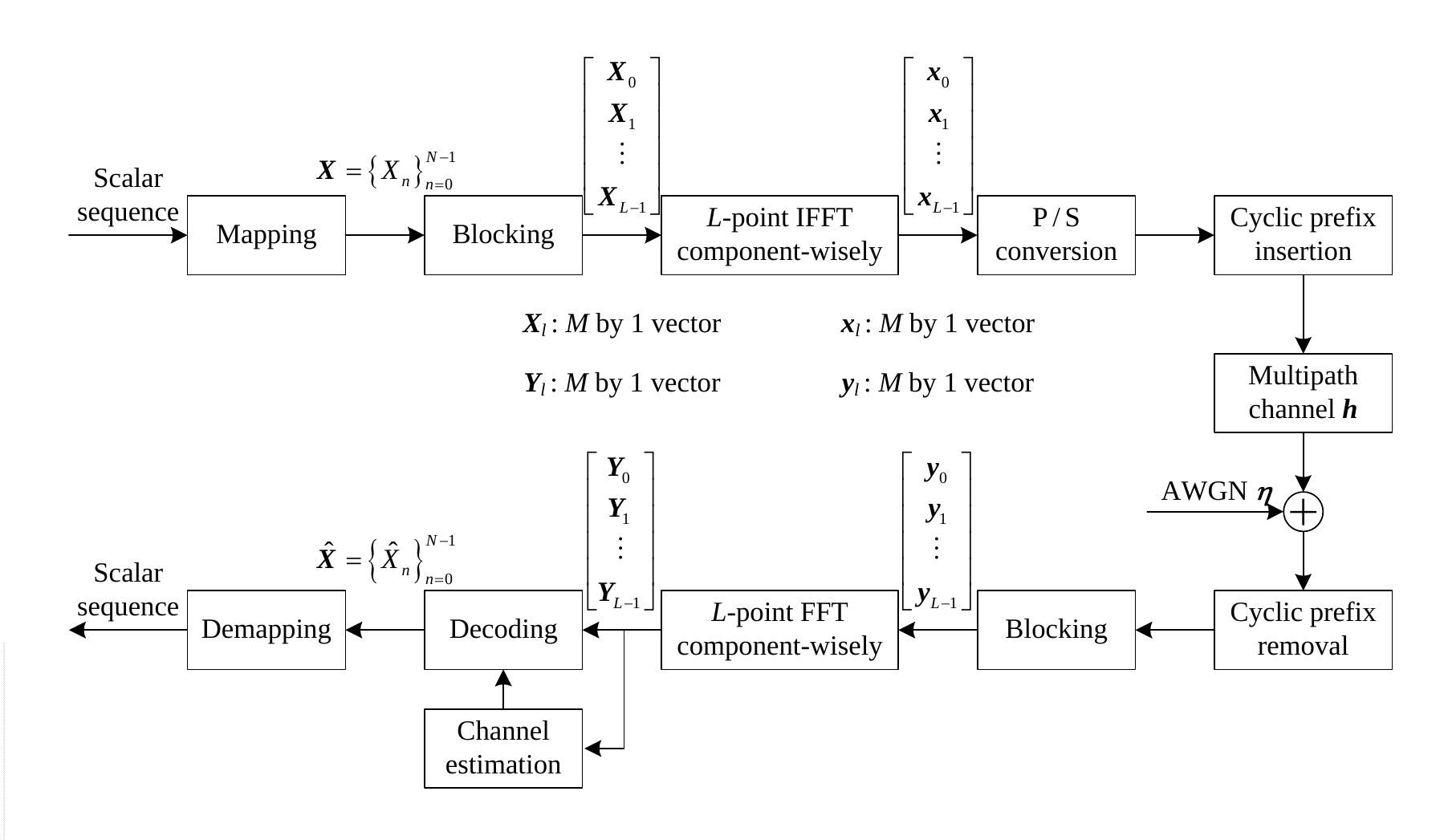}
\caption{Block diagram of vector OFDM modulation system.}
\end{figure}
\subsection{Vector OFDM Modulation}
In V-OFDM systems, $N$ symbols $\bm{X}=\left\{X_n\right\}_{n=0}^{N-1}$ are blocked into $L$ vectors (called vector blocks (VB)) of size $M$. Denote the $l$th transmitted VB in $\bm{X}$ as
\begin{equation}
\bm{X}_l=\left[X_{lM},X_{lM+1},\ldots,X_{lM+M-1}\right]^{\mathrm T},~l=0,1,\ldots,L-1
\end{equation}
where $(\cdot)^{\mathrm T}$ denotes the transpose. Assume the average power is normalized, i.e., $\mathbb{E}\left\{\left|X_n^2\right|\right\}=1,~n=0,1,\ldots,N-1$, where $\mathbb{E}\{\cdot\}$ denotes the mathematical expectation.

Accordingly, $\bm{x}_k$ is defined as the normalized VB-based IFFT of size $L$, i.e.,
\begin{equation}
\bm{x}_k=\frac{1}{\sqrt{L}}\sum\limits_{l=0}^{L-1}\bm{X}_l\mathrm e^{\mathrm j\frac{2\mathrm\pi}{L}kl},~k=0,1,\ldots,L-1.
\end{equation}
Here, $\bm{x}_k$ is a column vector of size $M$ represented as $\left[x_{kM},x_{kM+1},\ldots,x_{kM+M-1}\right]^{\mathrm T}$. After parallel to series (P/S) conversion, the transmitted signal sequence $\bm{x}=\left\{x_n\right\}_{n=0}^{N-1}$ is obtained as $\left[\bm{x}_0^{\mathrm T},\bm{x}_1^{\mathrm T},\ldots,\bm{x}_{L-1}^{\mathrm T}\right]^{\mathrm T}$. In order to avoid the interblock interference (IBI), the length of CP denoted by $\varGamma$ should not be shorter than the maximum time delay of a multipath channel. Note that $\varGamma$ does not need to be divisible by $M$. At the transmitter, the signal sequence $\bm{x}$ inserted by CP, is transmitted serially through the channel with the order $\left[x_{N-\varGamma},x_{N-\varGamma+1},\ldots,x_{N-1},x_0,x_1,\ldots,x_{N-1}\right]^{\mathrm T}$.

At the receiver, the received sequence is modeled by the transmitted signal through a frequency selective fading channel with complex AWGN. An inverse process as in transmitter is performed to recover the original symbols. After removing CP, the received sequence $\bm{y}=\left\{y_n\right\}_{n=0}^{N-1}$ is equal to the circular convolution of the transmitted signal and CIR with AWGN
\begin{equation}
y_n=x_n\circledast h_n+\xi_n,~n=0,1,\ldots,N-1
\end{equation}
where $\circledast$ denotes the circular convolution, CIR $\bm{h}=\left\{h_d\right\}_{d=0}^{D}$ and $D$ denotes the maximum time delay spread of the multipath channel. After zero padding of $\bm{h}$, CFR $\bm{H}=\left\{H_n\right\}_{n=0}^{N-1}$ is the $N$-point FFT (without normalization) of $\bm{h}$. Assume the additive noise $\bm{\xi}$ is independent and identically distributed (i.i.d.) random sequence whose entry $\xi_n\thicksim\mathcal{CN}\left(0,\sigma^2\right)$. Accordingly, define the transmitted signal-to-noise ratio (SNR) as $\rho=\frac{1}{\sigma^2}$. $\bm{y}$ is then blocked into $L$ column vectors of size $M$.
Denote the $k$th vector in $\bm{y}$ as
\begin{equation}
\bm{y}_k=\left[y_{kM},y_{kM+1},\ldots,y_{kM+M-1}\right]^{\mathrm T},~k=0,1,\ldots,L-1.
\end{equation}

Take the normalized component-wise vector FFT operation of size $L$ as
\begin{equation}
\bm{Y}_l=\frac{1}{\sqrt{L}}\sum\limits_{k=0}^{L-1}\bm{y}_k\mathrm e^{-\mathrm j\frac{2\mathrm\pi}{L}kl},~l=0,1,\ldots,L-1.
\end{equation}

The $l$th received VB in $\bm{Y}$ is also a column vector of size $M$ represented as $\bm{Y}_l=\left[Y_{lM},Y_{lM+1},\ldots,Y_{lM+M-1}\right]^{\mathrm T}$. It is derived from \cite{Xia2001} that the relationship between the $l$th transmitted VB $\bm{X}_l$ and received VB $\bm{Y}_l$ as
\begin{equation}
\bm{Y}_l=\bm{\mathcal{H}}_l\bm{X}_l+\bm{\Xi}_l,~l=0,1,\ldots,L-1
\end{equation}
where $\bm{\mathcal{H}}_l=\bm{\mathcal{H}}(z)\big|_{z=\mathrm e^{\mathrm j\frac{2\mathrm\pi}{L}l}}$ is a blocked channel matrix of the original ISI channel $H(z)$ as
\begin{equation}
\bm{\mathcal{H}}(z)=\begin{bmatrix}
\widetilde{H}_0(z)&z^{-1}\widetilde{H}_{M-1}(z)&\cdots&z^{-1}\widetilde{H}_1(z)\\
\widetilde{H}_1(z)&\widetilde{H}_0(z)&\cdots&z^{-1}\widetilde{H}_2(z)\\
\vdots&\vdots&\ddots&\vdots\\
\widetilde{H}_{M-2}(z)&\widetilde{H}_{M-3}(z)&\cdots&z^{-1}\widetilde{H}_{M-1}(z)\\
\widetilde{H}_{M-1}(z)&\widetilde{H}_{M-2}(z)&\cdots&\widetilde{H}_0(z)
\end{bmatrix}
\end{equation}
where $\widetilde{H}_m(z)=\sum\limits_kh_{kM+m}z^{-k}$ is the $m$th polyphase component of $H(z)$, $m=0,1,\ldots,M-1$. The additive noise $\bm{\Xi}_l$ in (6) is the blocked version of $\bm{\Xi}$ whose entries have the same power spectral density as in $\bm{\xi}$ that are i.i.d. complex Gaussian random variables.

Note that $\bm{\mathcal{H}}_l$ can be diagonalized as
\begin{equation}
\bm{\mathcal{H}}_l={\bf U}_l^{\mathrm H}{\bf H}_l{\bf U}_l
\end{equation}
where ${\bf U}_l$ is a unitary matrix whose entries $\left[{\bf U}\right]_{r,c}=\frac{1}{\sqrt{M}}\mathrm e^{-\mathrm j\frac{2\mathrm\pi}{N}(l+rL)c},~r,c=0,1,\ldots,M-1$, and $(\cdot)^{\mathrm H}$ denotes the conjugate transpose, ${\bf H}_l$ is an $M\times M$ diagonal matrix defined as
\begin{equation}
{\bf H}_l=\mathrm{diag}\left\{H_l,H_{l+L},\ldots,H_{l+(M-1)L}\right\}.
\end{equation}

Furthermore, the unitary matrix ${\bf U}_l$ can be rewritten as
\begin{equation}
{\bf U}_l={\bf F}_M{\bf \Lambda}_l
\end{equation}
where ${\bf F}_M$ denotes the $M\times M$ normalized discrete Fourier transform (DFT) matrix whose entries $\left[{\bf F}_M\right]_{r,c}=\frac{1}{\sqrt{M}}\mathrm e^{-\mathrm j\frac{2\mathrm\pi}{M}rc},~r,c=0,1,\ldots,M-1$, and ${\bf\Lambda}_l$ is a diagonal matrix defined as
\begin{equation}
{\bf\Lambda}_l=\mathrm{diag}\left\{1,\mathrm e^{-\mathrm j\frac{2\mathrm\pi}{N}l},\ldots,\mathrm e^{-\mathrm j\frac{2\mathrm\pi}{N}(M-1)l}\right\}.
\end{equation}

It can be seen from (6) that the original ISI channel $H(z)$ of $D+1$ symbols interfered together is converted to $L$ vector subchannels, each of which may have $M$ symbols interfered together. Note that $M$ is the vector size and can be flexibly designed. When $M=1$, (6) is back to the original OFDM, i.e., no ISI occurs in each subchannel. When $M=N$, all $D+1$ symbols are interfered together and it is back to the SC-FDE.
\subsection{Pilot Pattern}
Now, we rewrite the relationship of inputs and outputs in (6) for the better understanding of channel transmission structure
\begin{equation}
{\bf U}_l\bm{Y}_l={\bf H}_l{\bf U}_l\bm{X}_l+{\bf U}_l\bm{\Xi}_l,~l=0,1,\ldots,L-1.
\end{equation}

It is straightforward to show that after the unitary transformation, the $l$th VB $\bm{X}_l$ is transmitted parallel over the subchannels $H_l,H_{l+L},\ldots,H_{l+(M-1)L}$. Mathematically, ${\bf U}_l$ is a kind of rotation matrix, ${\bf H}_l$ can be thus viewed as the equivalent channel Fourier coefficients.

Denote $P$ as the number of pilot channels and assume $L$ is divisible by $P$. If the $l_p$th VB $\bm{X}_{l_p}=\left[X_{l_pM},X_{l_pM+1},\ldots,X_{l_pM+M-1}\right]^{\mathrm T}$ is allocated to transmit pilot symbols, then the pilot symbols are transmitted parallel over the equivalent channels $H_{l_p},H_{l_p+L},\ldots,H_{l_p+(M-1)L}$. Furthermore, if $P$ pilot channels are evenly distributed over $L$ subchannels, i.e.,
\begin{equation}
l_p=\frac{pL}{P},~p=0,1,\ldots,P-1
\end{equation}
then the equivalent channels allocated to transmit pilot symbols are aligned as
\begin{align}
\bm{H_P}=&\left[H_{l_0},H_{l_1},\ldots,H_{l_{P-1}}, H_{l_0+L},H_{l_1+L},\ldots,H_{l_{P-1}+L},\ldots\right.\nonumber\\
&\left.\ldots,H_{l_0+(M-1)L},H_{l_1+(M-1)L},\ldots,H_{l_{P-1}+(M-1)L}\right]^{\mathrm T}\nonumber\\
=&\big[H_0,H_{\frac{L}{P}},\ldots,H_{N-\frac{L}{P}}\big]^{\mathrm T}
\end{align}

Therefore, it is not difficult to find that pilot symbols are also evenly distributed over the equivalent channels. Fig. 2 shows the pilot pattern for V-OFDM systems with parameters $L=8,M=2,P=2$. Due to the even distribution, $\bm{H_P}$ can be regarded as the downsampling of the CFR $\bm H$ such that for most cases $D<MP$, we usually only need to perform IFFT with size $MP$ rather than $N$ to exactly estimate the CIR $\bm{h}$.
\begin{figure}[t]
\centering
\includegraphics[width=3.5in]{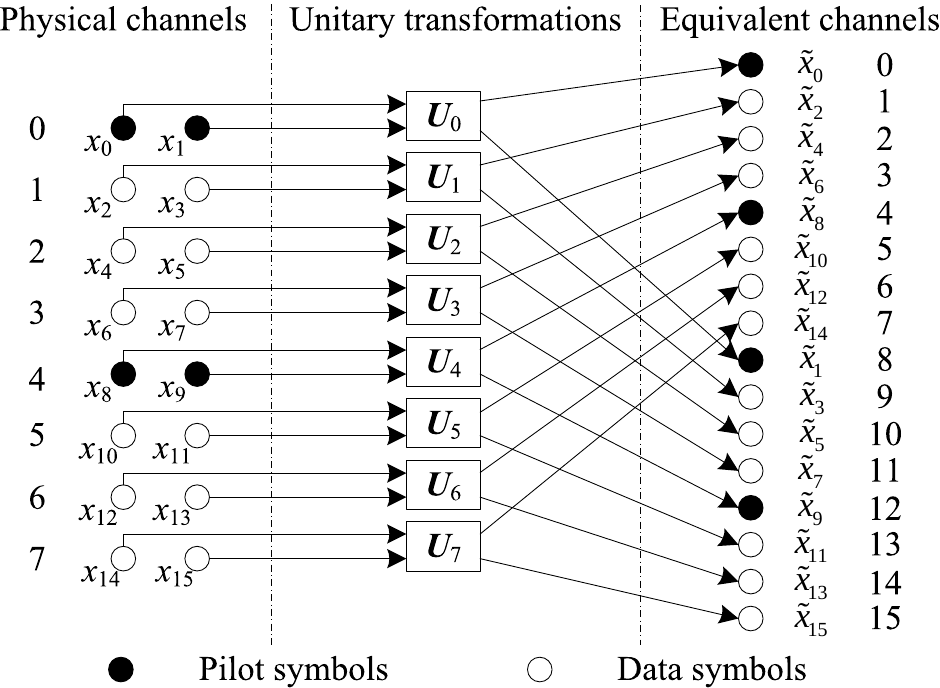}
\caption{Pilot pattern for vector OFDM system with $L=8,M=2,P=2$.}
\end{figure}
\section{Sparse IFFT for Channel Estimation}
With the increase of communication bandwidth, the number of equivalent channels $N$ needs to be increased and a signal sequence can be transmitted over more parallel channels simultaneously. Accordingly, either the number of VBs $L$ or the VB size $M$ increases proportionally. More parallel channels means higher rate for transmission, which however, increases the computational complexity of both channel estimation and decoding. In this section, an SIFFT-based approach is proposed for sparse multipath channel estimation that can directly obtain the nonzero channel coefficients and their corresponding coordinates.

Denote $\widetilde{\bm Y}_l={\bf U}_l\bm{Y}_l$, $\widetilde{\bm X}_l={\bf U}_l\bm{X}_l$, and $\widetilde{\bm \Xi}_l={\bf U}_l\bm{\Xi}_l$. Then (12) is further rewritten as
\begin{equation}
\widetilde{\bm Y}_l={\bf H}_l\widetilde{\bm X}_l+\widetilde{\bm \Xi}_l,~l=0,1,\ldots,L-1.
\end{equation}
Note that after the unitary transformation, the entries of $\widetilde{\bm \Xi}_l$ are also i.i.d. complex Gaussian random variables. $\bm H_l$ is defined as a column vector $\left[H_l,H_{l+L},\ldots,H_{l+(M-1)L}\right]^{\mathrm T}$ and $\widehat{\bm H}_l$ is the estimator of $\bm H_l$. It is convenient to estimate $\bm H_l$ by the least squares approach such that $\widehat{\bm H}_l=\big[\mathrm{diag}\big\{\widetilde{\bm X}_l\big\}\big]^{-1}\widetilde{\bm Y}_l$, $l=0,1,\ldots,L-1$.

Consider the pilot channels are evenly distributed over $L$ channels as (13). Denote $\widehat{\bm H}_{\bm P}$ and $\widehat{\bm h}$ as the estimators of $\bm{H_P}$ and $\bm{h}$, respectively. According to (14), $\bm H_{\bm P}$ is comprised of $\bm H_{l_p},~p=0,1,\ldots,P-1$, and can be estimated by $\widehat{\bm H}_{l_p}=\big[\mathrm{diag}\big\{\widetilde{\bm X}_{l_p}\big\}\big]^{-1}\widetilde{\bm Y}_{l_p}$, $p=0,1,\ldots,P-1$.

With the knowledge of column vector $\bm{H_P}$, we can further obtain column vector $\bm{h}$ by implementing the IFFT operation without normalization. Since the additive noise $\widetilde{\bm\Xi}_{l_p}$ is an $M\times 1$ i.i.d. random sequence whose entries $\big[\widetilde{\bm \Xi}_{l_p}\big]_m\thicksim\mathcal{CN}\left(0,\sigma^2\right)$, $m=0,1,\ldots,M-1$, it is not difficult to check that $\widehat{\bm H}_{l_p}$ is an unbiased estimator of $\bm H_{l_p}$, i.e., $\mathbb{E}\big\{\widehat{\bm H}_{l_p}\big\}=\bm H_{l_p}$, $p=0,1,\ldots,P-1$.
Besides, the mean squared error (MSE) of $\widehat{\bm H}_{l_p}$, denoted by $\mathbf\Sigma_{l_p}$, can be derived as
\begin{align}
\mathbf\Sigma_{l_p}&=\mathbb{E}\Big[\big(\widehat{\bm H}_{l_p}-\bm H_{l_p}\big)\big(\widehat{\bm H}_{l_p}-\bm H_{l_p}\big)^{\mathrm H}\Big]\nonumber\\
&=\sigma^2\Big[\mathrm{diag}\Big\{\big|\widetilde{\bm X}_{l_p}\big|^2\Big\}\Big]^{-1}
\end{align}

Since pilot signals are transmitted through a multipath channel with AWGN, our goal is to design the pilot symbols to minimize the MSE of estimator $\widehat{\bm h}$, i.e.,
\begin{equation}
\min\limits_{
\bm X_{l_p}\in \mathbb{X}^{M},~p=0,1,\ldots,P-1.
}\mathrm{tr}\Big\{\mathbb{E}\Big[\big(\widehat{\bm h}-\bm h\big)\big(\widehat{\bm h}-\bm h\big)^{\mathrm H}\Big]\Big\}
\end{equation}
where $\mathbb{X}$ denotes the constellation of the transmitted symbol $X_n,~n=0,1,\ldots,N-1$, $\mathrm{tr}\{\cdot\}$ is the trace of square matrix defined as the sum of the main diagonal. Obviously, $\widehat{\bm h}$ is also an unbiased estimator of $\bm h$ due to the linear transformation
\begin{align}
\mathbb{E}\big\{\widehat{\bm h}\big\}&=\frac{\mathbb{E}\big\{\mathbf F_{MP}^{-1}\widehat{\bm H}_{\bm P}\big\}}{\sqrt{MP}}=\frac{\mathbf F_{MP}^{-1}\mathbb{E}\big\{\widehat{\bm H}_{\bm P}\big\}}{\sqrt{MP}}\nonumber\\
&=\frac{\mathbf F_{MP}^{-1}\bm H_{\bm P}}{\sqrt{MP}}=\bm h
\end{align}
where $\mathbf F_{MP}^{-1}$ denotes the $MP\times MP$ normalized inverse DFT (IDFT) matrix whose entries $\big[\mathbf F_{MP}^{-1}\big]_{r,c}=\frac{1}{\sqrt{MP}}\mathrm e^{\mathrm j\frac{2\mathrm\pi}{MP}rc},~r,c=0,1,\ldots,MP-1$.

Now, we check the MSE of estimator $\widehat{\bm h}$ as
\begin{equation}
\mathbb{E}\Big[\big(\widehat{\bm h}-\bm h\big)\big(\widehat{\bm h}-\bm h\big)^{\mathrm H}\Big]=\frac{\mathbf F_{MP} \mathbf\Sigma\mathbf F_{MP}^{-1}}{MP}
\end{equation}
where $\mathbf\Sigma$ is a diagonal matrix whose diagonal entries $\left[\mathbf\Sigma\right]_{p+mP}=\left[\mathbf \Sigma_{l_p}\right]_m,~m=0,1,\ldots,M-1,~p=0,1,\ldots,P-1$. It is not difficult to find that the diagonal entries in (19) are the same, hence the entries of $\big(\widehat{\bm h}-\bm h\big)$ are identically distributed with complex Gaussian random variables, but may not be independent of each other. The trace of (19) can also be simplified as
\begin{equation}
\mathrm{tr}\Big\{\mathbb{E}\Big[\big(\widehat{\bm h}-\bm h\big)\big(\widehat{\bm h}-\bm h\big)^{\mathrm H}\Big]\Big\}=\frac{\mathrm{tr}\left\{\mathbf\Sigma\right\}}{MP}
\end{equation}
Note that (20) is only dependent on the sum of the MSE of each pilot channel. Therefore, (17) can be optimized by designing the pilot symbols for each channel, respectively, such that $\min\limits_{\bm X_{l_p}\in \mathbb{X}^{M}}\frac{\mathrm{tr}\left\{\mathbf\Sigma_{l_p}\right\}}{MP}$, $p=0,1,\ldots,P-1$.

Consider the total number of transmitted signals $N=1024$ is fixed, and $P$ pilot channels are evenly distributed over $L$ channels with $l_p=16p$. Assume the transmitted sequence $\bm X=\left\{x_n\right\}_{n=0}^{N-1}$ is binary phase-shift keying (BPSK) signals in V-OFDM system. For the $l_p$th pilot channel, the normalized expectation of MSE over $M$ symbols, defined as $\sigma_{l_p}^2=\frac{\mathrm{tr}\left\{\mathbf\Sigma_{l_p}\right\}}{M\sigma^2}$, should be minimized. According to Parseval's theorem, $\sigma_{l_p}^2\geqslant 1$. Note that such pilot symbol design is not unique. Table I lists a type of pilot symbol design for different $L$ and $M$. It is shown that such design can keep $\sigma_{l_p}^2$ within $1\thicksim2$, and does not increase with $M$. In Section III B, we will find that with such pilot symbol design, the SNR and the power ratio of the dominant entries to the rest entries in $\widehat{\bm h}$ are in the same order of magnitude, regardless of the parameters $M$ and $P$, which is accordance with approximately sparse channel.

\begin{table*}[t]
\centering
\caption{A Type of Pilot Symbols Design for BPSK Modulated Vector OFDM System}
\begin{subtable}[$L=256,M=4,N=1024,P=16$]{
\resizebox{\textwidth}{!}{
\begin{tabular}{||c|c|c||c|c|c||c|c|c||c|c|c||}
\hline
$p$ & $\bm{X}_{l_p}$ & $\sigma_{l_p}^{2}$ & $p$ & $\bm{X}_{l_p}$ & $\sigma_{l_p}^{2}$ & $p$ & $\bm{X}_{l_p}$ & $\sigma_{l_p}^{2}$ & $p$ & $\bm{X}_{l_p}$ & $\sigma_{l_p}^{2}$\\
\hline
$0$ & $~+1~-1~-1~-1~$ & $1.0000$ & $4$ & $~+1~-1~-1~-1~$ & $1.4118$ & $8$ & $~+1~-1~-1~-1~$ & $2.0000$ & $12$ & $~+1~-1~-1~-1~$ & $1.4118$\\
\hline
$1$ & $~+1~+1~+1~-1~$ & $1.0198$ & $5$ & $~-1~+1~-1~-1~$ & $1.6723$ & $9$ & $~-1~+1~-1~-1~$ & $1.9918$ & $13$ & $~+1~-1~-1~-1~$ & $1.2104$\\
\hline
$2$ & $~+1~-1~-1~-1~$ & $1.0848$ & $6$ & $~+1~-1~-1~-1~$ & $1.8986$ & $10$ & $~+1~-1~-1~-1~$ & $1.8986$ & $14$ & $~+1~-1~-1~-1~$ & $1.0848$\\
\hline
$3$ & $~-1~+1~-1~-1~$ & $1.2104$ & $7$ & $~+1~-1~-1~-1~$ & $1.9918$ & $11$ & $~+1~-1~-1~-1~$ & $1.6723$ & $15$ & $~+1~-1~-1~-1~$ & $1.0198$\\
\hline
\end{tabular}}}
\end{subtable}
\begin{subtable}[$L=128,M=8,N=1024,P=8$]{
\begin{tabular}{||c|c|c||c|c|c||}
\hline
$p$ & $\bm{X}_{l_p}$ & $\sigma_{l_p}^{2}$ & $p$ & $\bm{X}_{l_p}$ & $\sigma_{l_p}^{2}$\\
\hline
$0$ & $~+1~+1~-1~+1~-1~-1~-1~-1~$ & $1.3333$ & $4$ & $~-1~+1~+1~-1~-1~-1~-1~-1~$ & $1.4118$\\
\hline
$1$ & $~-1~-1~-1~+1~+1~-1~+1~-1~$ & $1.2921$ & $5$ & $~+1~-1~+1~+1~-1~-1~-1~-1~$ & $1.3513$\\
\hline
$2$ & $~-1~+1~-1~+1~+1~-1~-1~-1~$ & $1.2736$ & $6$ & $~+1~-1~+1~+1~-1~-1~-1~-1~$ & $1.2736$\\
\hline
$3$ & $~-1~+1~-1~+1~+1~-1~-1~-1~$ & $1.3513$ & $7$ & $~+1~-1~+1~+1~-1~-1~-1~-1~$ & $1.2921$\\
\hline
\end{tabular}}
\end{subtable}
\begin{subtable}[$L=64,M=16,N=1024,P=4$]{
\begin{tabular}{||c|c|c||}
\hline
$p$ & $\bm{X}_{l_p}$ & $\sigma_{l_p}^{2}$\\
\hline
$0$ & $~+1~+1~-1~-1~+1~+1~-1~-1~+1~-1~+1~-1~-1~-1~-1~-1~$ & $1.2589$\\
\hline
$1$ & $~+1~-1~+1~+1~+1~-1~+1~+1~-1~+1~+1~+1~+1~-1~-1~-1~$ & $1.2566$\\
\hline
$2$ & $~+1~-1~-1~+1~-1~+1~-1~-1~-1~+1~+1~-1~-1~-1~-1~-1~$ & $1.1974$\\
\hline
$3$ & $~-1~+1~-1~-1~+1~-1~+1~+1~+1~-1~-1~-1~+1~-1~-1~-1~$ & $1.2566$\\
\hline
\end{tabular}}
\end{subtable}
\end{table*}

For the conventional OFDM signal, the MSEs of the CFR estimator for all subchannels are the same. For the V-OFDM, however, the entries in $\widetilde{\bm X}_l$ may no longer be constant modulus since unitary transformation $\bm U_l$ is evolved in the original pilot symbols, and $\mathbf\Sigma_{l_p}$ varies from subchannel to subchannel. Similar to the channel spectral nulls in the OFDM systems, if symbol spectral nulls are existed in $\widetilde{\bm X}_l$, the overall estimation error in the $l$th subchannel may be very large. For the sparse channel that the pilot symbol are not well designed, the noise component may be comparable to the dominant component in the estimator $\widehat{\bm h}$, and in what follows, $\widehat{\bm h}$ may be obtained as a non-sparse channel.

For a very broadband channel, $PM$ may become very large. In this case, it becomes expensive to implement the large size IFFT directly. Inspired by the idea of SFFT proposed in \cite{Hassanieh2014May}, for a sparse multipath channel, we can perform SIFFT to estimate CIR from CFR. In what follows, we will focus on such sparse multipath channel estimation with and without AWGN in the transmission, respectively.
\subsection{Exactly Sparse Multipath Channel}
We first consider a simple noiseless sparse multipath channel, i.e., $\bm\xi=\bm0$, which is called an exactly sparse multipath channel. The pilot signals are transmitted through a sparse channel with only $K$ nonzero taps spread but without additive noise. If IFFT is performed to estimate CIR from CFR, in this case, it is no doubt that the estimator $\widehat{\bm h}=\frac{\mathbf F_{MP}^{-1}\widehat{\bm H}_{\bm P}}{\sqrt{MP}}$ is also a column vector of size $MP$ with only $K$ nonzero entries.

The SIFFT-based channel estimation for exactly sparse multipath channel is illustrated in Algorithm 1. The input $\widehat{\bm H}_{\bm P}$ is first permuted by $P_{\sigma,0,b}$, then multiplied by flat window filtering $\bm G_{B,\alpha,\varepsilon}$. After the permutation and filtering, the nonzero channel taps can be sampled at the interval $\frac{n}{B}$. Substituting the permutation operator $P_{\sigma,1,b}$ for $P_{\sigma,0,b}$ and repeating the above process, the nonzero coordinates can be recovered from the phase difference between these two permutations, and their corresponding values are obtained by permutation $P_{\sigma,0,b}$. After repeating $1+\log K$ times, one can eventually recover $\widehat{\bm h}$ with exact $K$-sparse. The algorithm includes three functions:
\begin{itemize}
\item\textsc{ExactlySparseIFFT:} Iterate \textsc{CoordinateValue} and update $\widehat{\bm h}$, repeat $1+\log K$ times and eventually find $\widehat{\bm h}$ with exact $K$-sparse.
\item\textsc{CoordinateValue:} Access to \textsc{HashToBins} and obtain nonzero coordinates and their corresponding values. It can find more than half of nonzero entries in $\widehat{\bm h}$ each time.
\item\textsc{HashToBins:} Permute $\widehat{\bm H}_{\bm P}$ and guarantee that nonzero entries in $\widehat{\bm h}$ are separated into different bins and then compute $B$-dimensional IFFT in $\mathcal O(B\log B)$, where $B$ denotes the number of bins and is set proportional to $K$.
\end{itemize}

Similar to \cite{Hassanieh2014May}, the proposed algorithm has two fundamental steps, i.e., permutation and flat window filtering. The purpose of permutation is to separate nonzero coefficients into different bins randomly. The design of filtering is a tradeoff between the filter flatness and the support of the window. Rather than the exactly sparse algorithm presented in \cite{Hassanieh2014May} that recovers the sparse signal with only a few nonzero Fourier coefficients, Algorithm 1 is aimed to estimate the CIR from the CFR by using the pilot symbols. The main differences between them are listed as follows.
\begin{itemize}
\item Instead of the permutation operator presented in \cite{Hassanieh2014May}, in this subsection, we redefine the permutation operator $P_{\sigma,a,b}$ as
\begin{equation}
P_{\sigma,a,b}(\bm X)_k=X_{\pi_{\sigma,a}(k)}\mathrm e^{\mathrm j\frac{2\mathrm\pi}{n}bk}
\end{equation}
where $\bm X=\left\{X_k\right\}_{k=0}^{n-1}$ is a discrete sequence in frequency domain with size $n$, and $\pi_{\sigma,a}(k)=(\sigma k-a)\bmod n$. Denote $p_{\sigma,a,b}(\bm x)$ as the IDFT of $P_{\sigma,a,b}(\bm X)$. It is not hard to derive that
\begin{equation}
p_{\sigma,a,b}(\bm x)_{\pi_{\sigma,b}(k)}=x_k\mathrm e^{\mathrm j\frac{2\mathrm\pi}{n}ak}
\end{equation}
where $\pi_{\sigma,b}(k)=(\sigma k-b)\bmod n$, and $\bm x$ is the IDFT of $\bm X$. Compared with the definition of permutation in \cite{Hassanieh2014May}, when computing the coordinates of nonzero channel coefficients in Algorithm 1, the proposed permutation can be recovered directly without the required dictionary.
\item Different from the exactly sparse algorithm in \cite{Hassanieh2014May} being only suitable for integers, we further expand the application to complex domain since any nonzero tap in sparse channel is a complex number. More specifically, denote the resolution $\delta$ as the minimum value that nonzero entries can be detected, if $\delta$ is set less than or equal to the minimum magnitude of nonzero channel taps, then all nonzero channel coefficients can be recovered with high probability. For the flat window with Gaussian filtering, the sample sequence should be collected with the length at least $\mathcal O(\frac{B}{\alpha}\log\frac{MP}{\varepsilon})$, where $\varepsilon=\frac{\delta}{4n^2\Delta}$, it may thus increase the sample sequence length if the sparse channel exists small nonzero taps. Therefore, the proposed algorithm can be suitable for the complex channel at the expense of a potential higher complexity.
\end{itemize}

The exactly sparse case in \cite{Hassanieh2014May} is the case when a signal has only a few nonzero Fourier coefficients. Accordingly, Algorithm 1 is suitable for the case that the estimator $\widehat{\bm h}$ has only a few nonzero taps. Therefore, it is straightforward to employ Algorithm 1 to recover $K$ nonzero entries in $\widehat{\bm h}$.

It has been proved in \cite{Hassanieh2014May} that the complexity of exactly sparse algorithm is $\mathcal O(K\log MP)$. Furthermore, if the length of symbol sequence $MP$ is sufficiently large such that $MP\geqslant\mathcal O(\frac{B}{\alpha}\log\frac{MP}{\varepsilon})$ is satisfied, Algorithm 1 can recover the correct coordinates and their corresponding values with high probability.
\begin{algorithm}[t]
\caption{Exactly Sparse Multipath Channel}
\begin{algorithmic}
\Require $\widehat{\bm H}_{\bm P},~K,~M,~P$
\Ensure $\widehat{\bm h}$
\Function{ExactlySparseIFFT}{$\widehat{\bm H}_{\bm P},K,M,P$}
\State \textbf{initialization:} $\widehat{\bm h}\gets0,~n\gets MP$
\For {$t\gets0,1,\dots,\log k$}
\State $k\gets \frac{K}{2^t},~\alpha\propto\frac{1}{2^t}$
\State $\widehat{\bm h}\gets\widehat{\bm h}+\textsc{CoordinateValue}(\widehat{\bm H}_{\bm P},\widehat{\bm h},k,n,\alpha)$
\EndFor
\State $\widehat{\bm h}\gets\mathop{\arg\max}\limits_{|\mathcal J|=K}\big\|\widehat{\bm h}_{\mathcal J}\big\|_2$\\
\Return $\widehat{\bm h}$
\EndFunction
\Function{CoordinateValue}{$\widehat{\bm H}_{\bm P},\widehat{\bm h},k,n,\alpha$}
\State $B\propto k$
\State $\varepsilon\gets \frac{\delta}{4n^2\Delta}$, for $\Delta\geqslant\max\big|\widehat{\bm h}\big|,~\delta\leqslant\min\big|\widehat{\bm h}\big|$
\State Choose $\sigma$ randomly from $\{1,3,\ldots,n-1\}$
\State Choose $b$ randomly from $\{0,1,\ldots,n-1\}$
\State $\bm w\gets\textsc{HashToBins}(\widehat{\bm H}_{\bm P},\widehat{\bm h},n,P_{\sigma,0,b},B,\alpha,\varepsilon)$
\State $\bm w'\gets\textsc{HashToBins}(\widehat{\bm H}_{\bm P},\widehat{\bm h},n,P_{\sigma,1,b},B,\alpha,\varepsilon)$
\State \textbf{initialization:} $\widehat{\bm h}\gets0$
\State $\mathcal{J}=\left\{j\big|\left|w_j\right|\geqslant\frac{\delta}{2}\right\}$
\ForAll {$j\in\mathcal{J}$}
\State $i\gets\mathrm{round}\big(\frac{n}{2\pi}\angle\frac{w'_j}{w_j}\big)\bmod n$
\State $\widehat h_i\gets w_j$
\EndFor\\
\Return $\widehat{\bm h}$
\EndFunction
\Function{HashToBins}{$\widehat{\bm H}_{\bm P},\widehat{\bm h},n,P_{\sigma,a,b},B,\alpha,\varepsilon$}
\State $\bm U\gets\bm G_{B,\alpha,\varepsilon}P_{\sigma,a,b}(\widehat{\bm H}_{\bm P})$
\For {$i\gets0,1,\ldots,B-1$}
\State $V_i\gets\sum\limits_jU_{i+Bj}$
\EndFor
\State $\bm v\gets\mathcal{F}^{-1}(\bm V)$
\For {$j\gets0,1,\ldots,B-1$}
\State $w_j\gets v_j-\big[\bm g_{B,\alpha,\varepsilon}*p_{\sigma,a,b}\big(\widehat{\bm h}\big)\big]_{\frac{n}{B}j}$
\EndFor\\
\Return $\bm w$
\EndFunction
\end{algorithmic}
\end{algorithm}
\subsection{Approximately Sparse Multipath Channel}
Now, we consider a more practical scenario that the pilot signals are transmitted through a sparse channel with only $K$ nonzero taps spread and AWGN, which is called an approximately sparse multipath channel. Since AWGN is induced during the transmission, the estimator $\widehat{\bm h}$ is no longer with only $K$ nonzero entries. In fact, the estimator $\widehat{\bm h}=\frac{\mathbf F_{MP}^{-1}\widehat{\bm H}_{\bm P}}{\sqrt{MP}}$ has $K$ dominant entries and the rest entries are small, when the SNR is not low. For the approximately sparse vector $\widehat{\bm h}$, define the parameter $\eta$ as the maximum expectation power ratio of the $K$ selected entries to the rest entries such that
\begin{equation}
\eta=\max_{|\mathcal J|=K}\mathbb E\left\{\frac{\big\|\widehat{\bm h}_{\mathcal J}\big\|_2^2}{\big\|\widehat{\bm h}-\widehat{\bm h}_{\mathcal J}\big\|_2^2}\right\}
\end{equation}
where $\|\cdot\|_2$ denotes the $\ell^2$ norm of a vector. $\eta$ reflects how approximately the sparse multipath channel is and determines the root-mean-square error (RMSE) of SIFFT algorithm. In particular, exactly sparse is an extreme case for $\eta\rightarrow\infty$.

The SIFFT-based channel estimation for approximately sparse multipath channel is shown in Algorithm 2 that has the following basic idea. To deal with noise, the algorithm estimates the nonzero coordinates and their corresponding values separately. For the coordinate estimation, all the coordinates are first divided into small regions. The input $\widehat{\bm H}_{\bm P}$ is permutated randomly by $P_{\sigma,a,b}$ and $P_{\sigma,a+\tau,b}$, respectively, then multiplied by flat window filtering $\bm G_{B,\alpha,\varepsilon}$. The phase difference between these two permutations determines the circular distance to each region. Select the appropriate regions with the nearest circular distance and get one vote. After repeating the above process $T_R$ times, choose the final regions with more than $\frac{T_R}{2}$ votes. By narrowing the regions of nonzero coordinates in each iteration, the algorithm eventually obtains the nonzero coordinates. For the value estimation, after the permutation and filtering, the nonzero values corresponding to the coordinates estimated before are obtained by permutation $P_{\sigma,a,b}$. Repeating $T_V$ times and choose the median as the estimations of the values such that the estimation error decreases exponentially with $T_V$. Repeat the above process $T_A$ times and ultimately recover $\widehat{\bm h}$ with $K$ dominant taps. The algorithm includes five functions, in which \textsc{HashToBins} is defined the same as in Algorithm 1.

\begin{itemize}
\item\textsc{ApproximatelySparseIFFT:} Iterate \textsc{Coordinate} and \textsc{Value}, then update $\widehat{\bm h}$. In each iteration, reduce $k$-sparse to $\frac{k}{4}$-sparse, repeat $T_A$ times and eventually find $\widehat{\bm h}$ with $K$ dominant entries.
\item\textsc{Coordinate:} Access to \textsc{Range} and narrow the range of dominant coordinates, repeat $T_C$ times until the dominant coordinates are uniquely determined.
\item\textsc{Range:} Permute $\widehat{\bm H}_{\bm P}$ randomly with $T_R$ times, divide all the coordinates into several regions, find the appropriate regions with the nearest circular distance and then gets one vote. After repeating $T_R$ times, choose the final regions with more than $\frac{T_R}{2}$ votes.
\item\textsc{Value:} Access to \textsc{HashToBins} and obtain the estimations of the values, repeat $T_V$ times and take the median of such values with real and imaginary parts, respectively.
\end{itemize}

\begin{algorithm}[t]
\caption{Approximately Sparse Multipath Channel}
\begin{algorithmic}
\Require $\widehat{\bm H}_{\bm P},~K,~M,~P$
\Ensure $\widehat{\bm h}$
\Function{ApproximatelySparseIFFT}{$\widehat{\bm H}_{\bm P},K,M,P$}
\State \textbf{initialization:} $\widehat{\bm h}\gets0,~n\gets MP,~\varepsilon\gets\frac{1}{4n^2}$
\State $T_A\propto\frac{\log K}{\log\log K}$
\For {$t\gets0,1,\ldots,T_A-1$}
\State $\alpha\propto\frac{1}{(t+1)^4},~B\propto\frac{K}{(t+1)^6},~k\propto K\prod\limits_{i=1,2,\ldots,t}\frac{1}{i^2}$
\State $T_V\propto\log(\frac{B}{k\alpha})$
\State $\mathcal L\gets\textsc{Coordinate}(\widehat{\bm H}_{\bm P},\widehat{\bm h},n,B,\alpha,\varepsilon)$
\State $\widehat{\bm h}\gets\widehat{\bm h}+\textsc{Value}(\widehat{\bm H}_{\bm P},\widehat{\bm h},3k,n,B,\varepsilon,\mathcal L,T_V)$
\EndFor
\State $\widehat{\bm h}\gets\mathop{\arg\max}\limits_{|\mathcal J|=K}\big\|\widehat{\bm h}_{\mathcal J}\big\|_2$\\
\Return $\widehat{\bm h}$
\EndFunction
\Function {Value}{$\widehat{\bm H}_{\bm P},\widehat{\bm h},k,n,B,\varepsilon,\mathcal L,T_V$}
\For {$t\gets0,1,\ldots,T_V-1$}
\State Choose $\sigma$ randomly from $\{1,3,\ldots,n-1\}$
\State Choose $a,b$ randomly from $\{0,1,\ldots,n-1\}$
\State $\bm w^{(t)}\gets\textsc{HashToBins}(\widehat{\bm H}_{\bm P},\widehat{\bm h},n,P_{\sigma,a,b},B,\varepsilon,\alpha)$
\EndFor
\State \textbf{initialization:} $\widehat{\bm h}\gets0$
\ForAll {$\ell\in\mathcal L$}
\State $\widehat h_\ell\gets\mathop{\mathrm{median}}\limits_{t\in\{1,2,\ldots,T_V\}}\left\{w_{\hbar_{\sigma,b}(i)}^{(t)}\mathrm e^{-\mathrm j\frac{2\mathrm\pi}{n}\sigma a\ell}\right\}$
\EndFor
\State $\widehat{\bm h}\gets\mathop{\arg\max}\limits_{|\mathcal J|=k}\big\|\widehat{\bm h}_{\mathcal J}\big\|_2$\\
\Return $\widehat{\bm h}$
\EndFunction
\Function{HashToBins}{$\widehat{\bm H}_{\bm P},\widehat{\bm h},n,P_{\sigma,a,b},B,\alpha,\varepsilon$}
\State $\bm U\gets\bm G_{B,\alpha,\varepsilon}P_{\sigma,a,b}(\widehat{\bm H}_{\bm P})$
\For {$i\gets0,1,\ldots,B-1$}
\State $V_i\gets\sum\limits_jU_{i+Bj}$
\EndFor
\State $\bm v\gets\mathcal{F}^{-1}(\bm V)$
\For {$j\gets0,1,\ldots,B-1$}
\State $w_j\gets v_j-\big[\bm g_{B,\alpha,\varepsilon}*p_{\sigma,a,b}\big(\widehat{\bm h}\big)\big]_{\frac{n}{B}j}$
\EndFor\\
\Return $\bm w$
\EndFunction
\end{algorithmic}
\end{algorithm}
\begin{algorithm}[t]
\begin{algorithmic}
\Function {Coordinate}{$\widehat{\bm H}_{\bm P},\widehat{\bm h},n,B,\alpha,\varepsilon$}
\State\textbf{initialization:} $\ell_i\gets\frac{n}{B}i$ for $i\in\{0,1,\ldots,B-1\}$
\State Choose $\sigma$ randomly from $\{1,3,\ldots,n-1\}$
\State Choose $b$ randomly from $\{0,1,\ldots,n-1\}$
\State $\lambda\gets\frac{n}{B},J\gets\log n,T_C\gets\big\lceil\log_{\frac{J}{4}}\lambda\big\rceil,T_R\gets\big\lceil\log\log n\big\rceil$
\For {$t\gets0,1,\ldots,T_C-1$}
\State $\bm\ell\gets\textsc{Range}(\widehat{\bm H}_{\bm P},\widehat{\bm h},n,B,\sigma,b,\alpha,\varepsilon,\bm\ell,\lambda\left(\frac{4}{J}\right)^t,J,T_R)$
\EndFor
\State $\mathcal L\gets\pi_{\sigma,b}^{-1}(\bm\ell)$\\
\Return $\mathcal L$
\EndFunction
\Function {Range}{$\widehat{\bm H}_{\bm P},\widehat{\bm h},n,B,\sigma,b,\alpha,\varepsilon,\bm\ell,\lambda,J,T_R$}
\State\textbf{initialization:} $\mu_{i,j}\gets0$ for $i\in\{0,1,\ldots,B-1\},~j\in\{0,1,\ldots,J-1\}$
\State $\nu\propto\alpha^{\frac{1}{3}}$
\For {$t\gets0,1,\ldots,T_R-1$}
\State Choose $a$ randomly from $\{0,1,\ldots,n-1\}$
\State Choose random variable $\tau$ evenly distributed from $\big\{\big\lceil\frac{nJ\nu}{4\lambda}\big\rceil,\big\lceil\frac{nJ\nu}{4\lambda}\big\rceil+1,\ldots,\big\lfloor\frac{nJ\nu}{2\lambda}\big\rfloor\big\}$
\State $\bm w\gets\textsc{HashToBins}(\widehat{\bm H}_{\bm P},\widehat{\bm h},n,P_{\sigma,a,b},B,\alpha,\varepsilon)$
\State $\bm w'\gets\textsc{HashToBins}(\widehat{\bm H}_{\bm P},\widehat{\bm h},n,P_{\sigma,a+\tau,b},B,\alpha,\varepsilon)$
\For {$i\gets0,1,\ldots,B-1$}
\For {$j\gets0,1,\ldots,J-1$}
\State $\theta_{i,j}\gets \frac{2\pi}{n}\left(\ell_i+\frac{2j+1}{2J}\lambda+\sigma b\right)\bmod n$
\If {$\min\big\{\pm\big(\tau\theta_{i,j}-\angle\frac{w'_i}{w_i}\big)\bmod 2\pi\big\}\leqslant\pi\nu$}
\State $\mu_{i,j}\gets\mu_{i,j}+1$
\EndIf
\EndFor
\EndFor
\EndFor
\For {$i\gets0,1,\ldots,B-1$}
\State $\mathcal J\gets\left\{j\big|\mu_{i,j}>\frac{T_R}{2}\right\}$
\If {$\mathcal J\neq\O$}
\State $\ell_i\gets\min\limits_{j\in\mathcal J}\left\{\ell_i+\frac{j}{J}\lambda\right\}$
\Else
\State $\ell_i\gets\varnothing$
\EndIf
\EndFor\\
\Return $\bm\ell$
\EndFunction
\end{algorithmic}
\end{algorithm}

Similar to \cite{Hassanieh2014May}, the permutation operator $P_{\sigma,a,b}$ in this subsection is defined as
\begin{equation}
P_{\sigma,a,b}(\bm X)_k=X_{\pi_{\sigma,a}(k)}\mathrm e^{\mathrm j\frac{2\mathrm\pi}{n}\sigma bk}
\end{equation}
where $\pi_{\sigma,a}(k)=\sigma(k-a)\bmod n$. Accordingly, the IDFT of $P_{\sigma,a,b}$ is derived as
\begin{equation}
p_{\sigma,a,b}(\bm x)_{\pi_{\sigma,b}(k)}=x_k\mathrm e^{\mathrm j\frac{2\mathrm\pi}{n}\sigma ak}
\end{equation}
where $\pi_{\sigma,b}(k)=\sigma(k-b)\bmod n$.

As we will see from the proof in Appendix A, Algorithm 2 is suitable for the case when the estimator $\widehat{\bm h}$ has a few dominant taps and thus it can be applied to recover $K$ dominant taps in $\widehat{\bm h}$.

For a sparse multipath channel with AWGN, the complexity of the approximately sparse algorithm is $\mathcal O(K\log MP\log\frac{MP}{K})$, which is more complicated than the exactly case, but still much simpler than the IFFT with $\mathcal O(MP\log MP)$ operations. If the condition $MP\geqslant\mathcal O(\frac{B}{\alpha}\log\frac{MP}{\varepsilon})$ holds, Algorithm 2 can estimate the coordinates and their corresponding values with low RMSE. In Section V, we will present some simulation results to show that the RMSE of channel estimation is influenced by $\eta$ and ultimately determined by both $\rho$ and the design of pilot symbols.

As a result, the SIFFT-based channel estimation algorithm can not only reduce computational complexity, but also return the nonzero channel coefficients and their corresponding coordinates directly, which is significant to the following PIS decoding process.
\section{Partial Intersection Sphere Decoding}
In V-OFDM systems, the performances of several common decoding approaches were analyzed in \cite{Xia2001,Cheng2011,Zhang2006,Han2010,Han2014,Li2012}. It was proved in \cite{Han2010,Han2014,Cheng2011} that the diversity order of the ML decoding is $\min\left\{M,D+1\right\}$. In \cite{Li2012}, it was shown that the diversity order of the minimum mean square error (MMSE) decoding can achieve $\min\left\{\left\lfloor M2^{-R}\right\rfloor,D\right\}+1$, where $R$ represents the spectrum efficiency in bits/symbol, while for the zero-forcing (ZF) decoding, the diversity order is $1$. For all the demodulations of the ML, ZF and MMSE, they need to obtain $\bm{H}=\left\{H_n\right\}_{n=0}^{N-1}$ which is computed by the $N$-point FFT of the zero padded $\bm h$. If $N$ is very large, the computational load is high. For a sparse channel, most entries in $\bm{\mathcal H}_l$ are zero and the nonzero entries are regularly placed. Based on this observation, it may be better to extract nonzero entries over each row and search all possible symbol sequences lying in a certain sphere of radius around the received signal. Hence, the complexity of searching such possible sequences is exponential to the number of nonzero entries in each row of $\bm{\mathcal H}_l$, which is much less than $M$ when $M$ is large as what is studied in this paper. In this section, a partial intersection sphere (PIS) decoding algorithm is proposed for a sparse multipath channel. Here, partial intersection means the intersection of the existed and the current nonzero coordinate sets. In each iteration, the algorithm only needs to compare the current symbol sequences corresponding to the coordinates belonging to the partial intersection with the existed ones.

The proposed PIS decoding algorithm is illustrated as Algorithm 3 and explained below in detail. Assume the sparse channel $\bm h$ has only $K$ nonzero taps with the maximum delay $D$. Denote $\mathcal J$ as the set of coordinates of nonzero taps for the sparse multipath channel $\bm h$, i.e.,
\begin{equation}
\mathcal J=\left\{j\big|j\in\{0,1,\ldots,D\},~h_j\neq0\right\}
\end{equation}
and the cardinality of set $\mathcal J$ is equal to $K$, i.e., $\left|\mathcal J\right|=K$. Considering the special structure of $\bm{\mathcal H}$, denote $\mathcal I$ as the set of the reminders of the nonzero channel coefficient coordinates modulo $M$, i.e.,
\begin{equation}
\mathcal I=\left\{i\big|\forall j\in\mathcal J,~i=j\bmod M\right\}
\end{equation}

Denote $\kappa$ as the cardinality of $\mathcal I$, i.e., $\kappa=\left|\mathcal I\right|$. Suppose $i_0,~i_1,\ldots,~i_{\kappa-1}$ are the $\kappa$ entries in $\mathcal I$ with the ascending order $0\leqslant i_0<i_1<\cdots<i_{\kappa-1}\leqslant M-1$. For the case in this paper, we have $\kappa\leqslant K\ll M$. In what follows, it will be found that the diversity order for the PIS decoding is only related to $\kappa$.

In the $m$th iteration with $0\leqslant m\leqslant M-1$, denote $\mathcal U^{(m)}$ and $\mathcal V^{(m)}$ as the sets of the existed coordinates of the nonzero entries in the first $m-1$ rows and the current coordinates of the nonzero entries in the $m$th row of $\bm{\mathcal H}_l$, respectively. $\mathcal W^{(m)}$ (called partial intersection) is defined as the intersection of $\mathcal U^{(m)}$ and $\mathcal V^{(m)}$, i.e., $\mathcal W^{(m)}=\mathcal U^{(m)}\bigcap\mathcal V^{(m)}$. Note that $\bm{\mathcal H}_l$ is from a pseudo-circulant matrix (7) where the number of nonzero entries in each row is equal to $\kappa$. Recall that $\mathbb{X}$ is the constellation of the transmitted symbol $X_n$. For the initialization, the set of the existed coordinates of the nonzero entries $\mathcal U^{(0)}$ and the set of entire symbol sequences $\mathcal X^{(0)}$ are empty sets, respectively, i.e., $\mathcal U^{(0)}=\O$, $\mathcal X^{(0)}=\O$. Then, we describe the updating process of PIS decoding in the $m$th iteration with $0\leqslant m\leqslant M-1$ as follows.
\begin{enumerate}
\item Extract $\kappa$ entries from the $m$th row and $(m-i_0)\bmod M$th, $(m-i_1)\bmod M$th,$~\ldots~$, $(m-i_{\kappa-1})\bmod M$th columns of $\bm{\mathcal H}_l$, generate $\bm{\mathcal H}_l^{(m)}$ as $1\times\kappa$ vector $\big[\left[\bm{\mathcal H}_l\right]_{m,(m-i_0)\bmod M},\left[\bm{\mathcal H}_l\right]_{m,(m-i_1)\bmod M},\ldots$ $\ldots,\left[\bm{\mathcal H}_l\right]_{m,(m-i_{\kappa-1})\bmod M}\big]$.
\item Search all possible symbol sequences $\bm S=[S_0,S_1,\ldots,S_{\kappa-1}]^{\mathrm T},~\bm S\in\mathbb X^{\kappa}$, that lie in the certain sphere of radius $r$ around the received signal $\bm Y_l^{(m)}$ and generate the set of symbol sequences $\mathcal S^{(m)}$ as
\begin{equation}
\mathcal S^{(m)}=\Big\{\bm S\Big|\bm S\in\mathbb X^{\kappa},~\big|\bm Y_l^{(m)}-\bm{\mathcal{H}}_l^{(m)}\bm S\big|\leqslant r\Big\}
\end{equation}
where $\bm Y_l^{(m)}$ is the $m$th entry of the column vector $\bm Y_l$.
\item For each $\bm S\in\mathcal S^{(m)}$, construct an injective mapping of coordinates $f:~k\rightarrow(m-i_k)\bmod M,~k\in\{0,1,\dots,\kappa-1\}$. For each symbol sequence $\bm X^{(m)}\in\mathcal X^{(m)}$, where $\mathcal X^{(m)}$ is the set of entire symbol sequences generated from the previous iteration, compare the current symbol sequence $\bm S$ for the coordinates belonging to the partial intersection $\mathcal W^{(m)}$ with the existed symbol sequence $\bm X^{(m)}$, namely, if $X_w^{(m)}=S_{f^{-1}(w)}$ holds for all $w\in\mathcal W^{(m)}$, where $X_w^{(m)}$ stands for the $w$th entry in $\bm X^{(m)}$, then $\bm X^{(m)}$ is put into the set of symbol sequences $\mathcal X_{\bm S}^{(m)}$, which can be expressed as $\mathcal X_{\bm S}^{(m)}=\big\{\bm X^{(m)}\big|\bm X^{(m)}\in\mathcal X^{(m)},~\forall w\in\mathcal W^{(m)},~X_w^{(m)}\equiv S_{f^{-1}(w)}\big\}$. Then, insert the symbols whose coordinates belong to the complement of the partial intersection $\mathcal W^{(m)}$ to each symbol sequence $\bm X^{(m)}$, i.e., $\forall v\in\complement_{\mathcal V^{(m)}}\mathcal W^{(m)}$, where $\complement_{\mathcal V^{(m)}}\mathcal W^{(m)}$ stands for the complement of $\mathcal W^{(m)}$ in $\mathcal V^{(m)}$, set $X_v^{(m+1)}=S_{f^{-1}(v)}$, insert $X_v^{(m+1)}$ to each symbol sequence $\bm X^{(m)}$ in $\mathcal X_{\bm S}^{(m)}$ and generate $\bm X^{(m+1)}$, the new set of symbol sequences is thus updated as $\mathcal X_{\bm S}^{(m+1)}=\big\{\bm X^{(m+1)}\big|\bm X^{(m)}\in\mathcal X_{\bm S}^{(m)},~\forall u\in\mathcal U^{(m)},~X_u^{(m+1)}=X_u^{(m)};~\forall v\in\complement_{\mathcal V^{(m)}}\mathcal W^{(m)},~X_v^{(m+1)}=S_{f^{-1}(v)}\big\}$.
\item Repeat Step 3 by enumerating all $\bm S\in\mathcal S^{(m)}$. Then the set of entire symbol sequences $\mathcal X^{(m+1)}$ is obtained by the union of all $\mathcal X_{\bm S}^{(m+1)}$, i.e., $\mathcal X^{(m+1)}=\bigcup\limits_{\bm S\in\mathcal S^{(m)}}\mathcal X_{\bm S}^{(m+1)}$. $\mathcal U^{(m+1)}$ is updated to $\mathcal U^{(m)}\bigcup\mathcal V^{(m)}$ as the existed coordinates of nonzero entries for the next iteration.
\end{enumerate}

After $M$ iterations, the set of possible VB sequences $\mathcal X^{(M)}$ can be ultimately obtained, then choose the symbol sequence $\bm X^{(M)}\in\mathcal X^{(M)}$ with the minimum $\ell^2$ distance of $\big\|\bm Y_l-\bm{\mathcal{H}}_l\bm X^{(M)}\big\|_2$ as the estimation of the transmitted VB $\bm X_l$.
\begin{algorithm}[t]
\caption{Partial Intersection Sphere Decoding}
\begin{algorithmic}
\Require $\bm Y,~\bm h,~D,~L,~M,~r$
\Ensure $\widehat{\bm X}$
\Function{SparsePIS}{$\bm Y,\bm h,D,L,M,r$}
\State \textbf{initialization:} $\mathcal U^{(0)}\gets\O,~\mathcal X^{(0)}\gets\O$
\State Calculate $\bm{\mathcal{H}}_l,~l\in\{0,1,\ldots,L-1\}$ according to (7)
\State $\mathcal J\gets\left\{j\big|j\in\{0,1,\ldots,D\},~h_j\neq0\right\}$
\State $\mathcal I\gets\left\{i\big|\forall j\in\mathcal J,~i\gets j\bmod M\right\},~\kappa\gets|\mathcal I|$
\State $i_0,~i_1,\ldots,~i_{\kappa-1}$ are $\kappa$ entries in $\mathcal I$ with the ascending order $0\leqslant i_0<i_1<\cdots<i_{\kappa-1}\leqslant M-1$
\For {$l\gets0,1,\ldots,L-1$}
\For {$m\gets0,1,\ldots,M-1$}
\State $\bm Y_l^{(m)}$ is the $m$th entry of column vector $\bm Y_l$
\State $\bm{\mathcal{H}}_l^{(m)}$ is $1\times\kappa$ vector aligned as the $m$th row and $(m-i_0)\bmod M$th, $(m-i_1)\bmod M$th$,\ldots,(m-i_{\kappa-1})\bmod M$th columns of $\bm{\mathcal{H}}_l$
\State $\mathcal S^{(m)}\gets\left\{\bm S\Big|\bm S\in\mathbb X^{\kappa},~\big|\bm Y_l^{(m)}-\bm{\mathcal{H}}_l^{(m)}\bm S\big|\leqslant r\right\}$
\State $\mathcal V^{(m)}\gets
\left\{\iota^{(m)}\big|\forall i\in\mathcal I,~\iota^{(m)}\gets(m-i)\bmod M\right\}$
\State $\mathcal W^{(m)}\gets\mathcal U^{(m)}\bigcap\mathcal V^{(m)}$
\State $f:~k\rightarrow(m-i_k)\bmod M,~k\in\{0,1,\dots,\kappa-1\}$
\ForAll {$\bm S\in\mathcal S^{(m)}$}
\State $\mathcal X_{\bm S}^{(m)}\gets\big\{\bm X^{(m)}\big|\bm X^{(m)}\in\mathcal X^{(m)},~\forall w\in\mathcal W^{(m)},~X_w^{(m)}=S_{f^{-1}(w)}\big\}$
\State $\mathcal X_{\bm S}^{(m+1)}\gets\big\{\bm X^{(m+1)}\big|\bm X^{(m)}\in\mathcal X_{\bm S}^{(m)},~\forall u\in\mathcal U^{(m)},~X_u^{(m+1)}\gets X_u^{(m)};~\forall v\in\complement_{\mathcal V^{(m)}}\mathcal W^{(m)},~X_v^{(m+1)}\gets S_{f^{-1}(v)}\big\}$
\EndFor
\State $\mathcal X^{(m+1)}\gets\bigcup\limits_{\bm S\in\mathcal S^{(m)}}\mathcal X_{\bm S}^{(m+1)}$
\State $\mathcal U^{(m+1)}\gets\mathcal U^{(m)}\bigcup\mathcal V^{(m)}$
\EndFor
\State $\widehat{\bm X_l}\gets\mathop{\arg\min}\limits_{\bm X^{(M)}\in\mathcal X^{(M)}}\big\|\bm Y_l-\bm{\mathcal{H}}_l\bm X^{(M)}\big\|_2$
\EndFor\\
\Return $\widehat{\bm X}$
\EndFunction
\end{algorithmic}
\end{algorithm}

For a V-OFDM system with the PIS decoding, assume the CIR $\bm h$ and the average power of complex AWGN $\sigma^2$ are known at the receiver. Denote $\bm S_{\dagger}^{(m)}$ as the correct symbol sequence corresponding to the transmitted symbols, i.e., $\bm S_{\dagger}^{(m)}$ is extracted from the $(m-i_0)\bmod M$th, $(m-i_1)\bmod M$th,$~\ldots~$, $(m-i_{\kappa-1})\bmod M$th entries of $\bm X_l$ and generated as $\big[\left[\bm X_l\right]_{(m-i_0)\bmod M},\left[\bm X_l\right]_{(m-i_1)\bmod M},\ldots$ $\ldots,\left[\bm X_l\right]_{(m-i_{\kappa-1})\bmod M}\big]$. Hence, the distance between $\bm Y_l^{(m)}$ and $\bm{\mathcal{H}}_l^{(m)}\bm S_{\dagger}^{(m)}$, i.e., $\big|\bm Y_l^{(m)}-\bm{\mathcal{H}}_l^{(m)}\bm S_{\dagger}^{(m)}\big|$, is Rayleigh distributed with mean $\frac{\sqrt{\pi}}{2}\sigma$ and variance $\frac{4-\pi}{4}\sigma^2$. According to the cumulative distribution function of Rayleigh distribution, the probability that the transmitted symbol $\bm S_{\dagger}^{(m)}$ lies in the sphere radius $r$ in the $m$th iteration is
\begin{equation}
\Pr\big\{\big|\bm Y_l^{(m)}-\bm{\mathcal{H}}_l^{(m)}\bm S_{\dagger}^{(m)}\big|\leqslant r\big\}=1-\mathrm e^{-\frac{r^2}{\sigma^2}}
\end{equation}

From the previous analysis, the additive noise $\bm\Xi_l$ in (6) is an $M\times1$ vector whose entries are i.i.d. complex Gaussian random variables. Hence, for $m=0,1,\ldots,M-1$, the events that the transmitted symbol $\bm S_{\dagger}^{(m)}$ lies in the sphere are independent and the probabilities that each event occurs are the same. After $M$ iterations, the occurrence of event $\bm X_l\in\mathcal X^{(M)}$ is equivalent to the occurrence of all the $M$ events $\big|\bm Y_l^{(m)}-\bm{\mathcal{H}}_l^{(m)}\bm S_{\dagger}^{(m)}\big|\leqslant r,~m=0,1,\ldots,M-1$, then we have
\begin{equation}
\Pr\left\{\bm X_l\in\mathcal X^{(M)}\right\}=\big(1-\mathrm e^{-\frac{r^2}{\sigma^2}}\big)^M
\end{equation}

Note that $\mathcal X^{(M)}\subseteq\mathbb X^{M}$ is a set of possible VB sequences whose corresponding $\kappa$-dimensional symbol sequences lie in the sphere radius such that (28) holds for all $M$ rows in (6). In fact, the choice of sphere radius $r$ is a tradeoff between the symbol error rate (SER) performance and the computational complexity. With the increase of $r$, $\Pr\left\{\bm X_l\in\mathcal X^{(M)}\right\}$ also increases which consequently improves the SER performance. However, this means that more possible symbol sequences need to be compared in each iteration. The SER of the proposed PIS decoding $P_{\mathrm{PIS}}(r)$ is a function of sphere radius $r$ and calculated by the law of total probability
\begin{align}
P_{\mathrm{PIS}}(r)&=\Pr\big\{\widehat{\bm X_l}\neq\bm X_l\big|\bm X_l\in\mathcal X^{(M)}\big\}\Pr\left\{\bm X_l\in\mathcal X^{(M)}\right\}\nonumber\\
&+\Pr\big\{\widehat{\bm X_l}\neq\bm X_l\big|\bm X_l\notin\mathcal X^{(M)}\big\}\Pr\left\{\bm X_l\notin\mathcal X^{(M)}\right\}
\end{align}
Note that for $\bm X_l\notin\mathcal X^{(M)}$, there is no doubt that $\Pr\big\{\widehat{\bm X_l}\neq\bm X_l\big|\bm X_l\notin\mathcal X^{(M)}\big\}=1$ since $\widehat{\bm X_l}$ is chosen from $\mathcal X^{(M)}$. Denote $P_{\mathcal X^{(M)}}$ as the probability that symbol error occurs conditioned on $\bm X_l\in\mathcal X^{(M)}$, i.e.,
\begin{equation}
P_{\mathcal X^{(M)}}=\Pr\big\{\widehat{\bm X_l}\neq\bm X_l\big|\bm X_l\in\mathcal X^{(M)}\big\}
\end{equation}
Substituting (32) into (31), $P_{\mathrm{PIS}}(r)$ can be further simplified as
\begin{equation}
P_{\mathrm{PIS}}(r)=\big(1-\mathrm e^{-\frac{r^2}{\sigma^2}}\big)^MP_{\mathcal X^{(M)}}+1-\big(1-\mathrm e^{-\frac{r^2}{\sigma^2}}\big)^M
\end{equation}

For a V-OFDM system, the signal vector $\bm X_l$ needs to be specifically rotated/transformed to achieve full diversity for the ML decoding as done in \cite{Han2010}. In Section II B, it was analyzed that if $P$ subchannels with the indices $0,\frac{L}{P},\ldots,L-\frac{L}{P}$ are allocated to transmit pilot symbols, the IFFT/SIFFT-based channel estimation can be applied to recover CIR $\bm h$. For the remaining subchannels allocated to transmit data symbols, it is proved in Appendix B that the diversity order of the ML decoding for a sparse multipath channel is $\kappa$. We describe the diversity order by the exponential equality $P_{\mathrm{ML}}\doteq\rho^{-\kappa}$, which is mathematically defined as $\lim\limits_{\rho\rightarrow\infty}\frac{\ln P_{\mathrm{ML}}}{\ln\rho}=-\kappa$ \cite{Zheng2003}, where the SER of the ML decoding $P_{\mathrm{ML}}$ can be found in \cite{Xia2001,Han2010,Cheng2011}.

Instead of the ML decoding that enumerates symbol constellation and estimates the transmitted symbols with the minimum distance, the PIS decoding first generates the set of possible transmitted symbols $\mathcal X^{(M)}$ and then chooses the transmitted symbols only from $\mathcal X^{(M)}$ with the minimum distance. It is proved in Appendix C that $P_{\mathcal X^{(M)}}\leqslant P_{\mathrm{ML}}$. Accordingly, for sufficiently large $\rho$, $P_{\mathcal X^{(M)}}$ is exponentially less than or equal to $\rho^{-\kappa}$, which can be expressed as $P_{\mathcal X^{(M)}}~\dot\leqslant~\rho^{-\kappa}$, i.e., $\lim\limits_{\rho\rightarrow\infty}\frac{\ln P_{\mathcal X^{(M)}}}{\ln\rho}\leqslant-\kappa$. Therefore, from (33), $P_{\mathrm{PIS}}(r)$ is exponentially less than or equal to $\rho^{-\kappa}$ if and only if $1-\big(1-\mathrm e^{-\frac{r^2}{\sigma^2}}\big)^M$ is exponentially less than or equal to $\rho^{-\kappa}$, i.e.,
\begin{equation}
P_{\mathrm{PIS}}(r)~\dot\leqslant~\rho^{-\kappa}\Longleftrightarrow1-\big(1-\mathrm e^{-\frac{r^2}{\sigma^2}}\big)^M~\dot\leqslant~\rho^{-\kappa}
\end{equation}

We say the sphere radius $r$ is the asymptotically greater than or equal to the sphere radius $r'$, denoted by $r\succcurlyeq r'$, when
\begin{equation}
\lim_{\rho\rightarrow\infty}\frac{\ln P_{\mathrm{PIS}}(r)}{\ln\rho}\leqslant\lim_{\rho\rightarrow\infty}\frac{\ln P_{\mathrm{PIS}}(r')}{\ln\rho}
\end{equation}

For a sufficiently large $\rho=\frac{1}{\sigma^2}$, the infinitesimal $1-\big(1-\mathrm e^{-\frac{r^2}{\sigma^2}}\big)^M$ approximates to $M\mathrm e^{-\frac{r^2}{\sigma^2}}$. Substituting (34) into (35) and supposing $\rho$ is sufficiently large, we have $r\geqslant\sigma\sqrt{\ln M-2\kappa\ln\sigma}$. Furthermore, for a sufficiently large $\rho$, the term $\ln M$ can be neglected compared with $-2\kappa\ln\sigma$, then the necessary and sufficient condition of $P_{\mathrm{PIS}}(r)~\dot\leqslant~\rho^{-\kappa}$ is
\begin{equation}
r\succcurlyeq\sqrt{\frac{\kappa}{\rho}\ln\rho}
\end{equation}

It is well known that the SER of the proposed PIS decoding can not be better than that of the ML decoding. According to (36), we have the following lemma that gives the criterion of sphere radius satisfying $P_{\mathrm{PIS}}\doteq\rho^{-\kappa}$.
\begin{lemma}
The SER of the proposed PIS decoding is exponentially equal to that of the ML decoding in the choice of sphere radius $r\succcurlyeq\sqrt{\frac{\kappa}{\rho}\ln\rho}$ for a sufficiently large $\rho$.
\end{lemma}

For the V-OFDM system, it is known that the complexities with respect to complex multiplication operation of MMSE decoding and ML decoding are $\mathcal O(LM\log M+LM2^R)$ and $\mathcal O(LM^22^{RM})$, respectively. The PIS decoding only needs $M2^{R\kappa}$ trials with $\kappa$ complex multiplication operation in each trial. Hence, the complexity with respect to complex multiplication operation of the PIS decoding is $\mathcal O(\kappa LM2^{R\kappa})$. Besides, the evaluation and comparison operations should be taken into account in PIS decoding, which in fact, may vary from $\mathcal O(\kappa LM)$ to $\mathcal O(\kappa LM2^{RM})$ and are related to the cardinality of $\mathcal X^{(m)}$ and ultimately dependent of sphere radius $r$. As illustrated in Algorithm 3, the evaluation operation is an operator used for assignment where the source $S_{f^{-1}(v)}$ is a complex number and the destination $X_v^{(m)}$ is the $v$th entry in the symbol sequence $X^{(v)}$, i.e., $X_v^{(m)}=S_{f^{-1}(v)}$, while the comparison operation is one of relational operator used to check the equality of two complex numbers $X_w^{(m)}$ and $S_{f^{-1}(w)}$, i.e., $X_w^{(m)}\equiv S_{f^{-1}(w)}$, if the equality holds return $1$, otherwise return $0$. In assembly language, an evaluation operation or a comparison operation usually executes $1$ instruction cycle, whereas a real multiplication operation executes $4$ instruction cycles or slightly more due to hardware. Although different operations may have different execute time, the number of instruction cycles for any operation is fixed and can be seen as a constant. Then, the total complexity depends ultimately on the number of operations executed in the program. With the increase of $\rho$, it is wise to decrease the sphere radius $r$ such that the computational complexities with respect to the evaluation and comparison operations can reach the lower bound $\mathcal O(\kappa LM)$ with probability $1$. Note that $\lim\limits_{\rho\rightarrow\infty}\sqrt{\frac{\kappa}{\rho}\ln\rho}=0$, then we have the following theorem.
\begin{theorem}
For any given small sphere radius, the PIS decoding can achieve the diversity order $\kappa$ which is the same as the ML decoding for a sparse multipath channel, but the computational complexity decreases from $\mathcal O(LM^22^{RM})$ to $\mathcal O(\kappa LM2^{R\kappa})$ with probability $1$.
\end{theorem}
\begin{proof}
See Appendix D for the proof.
\end{proof}

Therefore, by choosing $r$ asymptotically equal to $\sqrt{\frac{\kappa}{\rho}\ln\rho}$, the proposed PIS decoding algorithm can balance the tradeoff between the SER performance and the computational complexity. Since the diversity order is $\kappa$ that depends on the set of the reminders of the $K$ nonzero channel coefficient coordinates modulo $M$, in practice, for a given channel model, i.e., for a given set of coordinates of nonzero channel coefficients, one may properly choose $M$ such that $\kappa$ is maximized.

\section{Simulation Results}
In this section, we provide simulation results to verify the previous analysis. The BPSK modulation is employed in the V-OFDM system. Sparse multipath channel $\bm h$ is modelled as $K$ i.i.d. complex Gaussian distributed nonzero taps $h_j\thicksim\mathcal{CN}\left(0,1\right),~j\in\mathcal J$ randomly distributed within the maximum delay $D$. We first employ the RMSE to evaluate the performances of the SIFFT-based sparse channel estimation. Then, we give an example of $6$ different channels with deterministic nonzero coordinates to make a comparison of the diversity order. Besides, we investigate the relationship between the BER performance of PIS decoding and the parameters $D$, $K$, $M$, respectively. Furthermore, the PIS decoding is compared with the conventional ZF, MMSE, ML decoding schemes in the V-OFDM system. Finally, channel estimation and decoding algorithm are jointly considered to show the BER performances in both OFDM and V-OFDM systems.

Figs. 3 and 4 show the RMSE performances of the SIFFT-based sparse channel estimation with and without noise, respectively. In the simulation of the SIFFT-based exactly sparse multipath channel estimation, the parameters $\alpha$ and $B$ are set to $\frac{K}{2^{t+4}}$ such that $\frac{B}{\alpha}\log\frac{MP}{\varepsilon}$ is a constant regardless of $K$. It can be seen from Fig. 3 that the RMSE of the channel estimation is below $0.02$ but reduces the complexity to $\mathcal O(K\log MP)$, where $P$ is pilot channel number. The estimation error is mainly caused by the imperfect permutation that the nonzero entries are not separated into different bins. For the SIFFT-based approximately sparse multipath channel estimation, suppose the parameters $\alpha=\frac{1}{(t+1)^4}$ and $B=\frac{K}{(t+1)^6}$ that can keep the collision at a relatively low level. Fig. 4 indicates that with the increase of $\rho$, $K$ dominant entries are slightly influenced by the rest entries $\eta$, which consequently, reduces the RMSE of channel estimation. For instance, when $K=4$, the RMSE of sparse channel estimation is below $0.01$. Whereas when $K=16$, however, there is a sharp decrease for $MP\leqslant18$ since the condition $MP\geqslant\mathcal O(\frac{B}{\alpha}\log\frac{MP}{\varepsilon})$ does not hold in this case. The complexity for the SIFFT-based approximately sparse multipath channel estimation is $\mathcal O(K\log MP\log\frac{MP}{K})$.
\begin{figure}[t]
\centering
\includegraphics[width=3.5in]{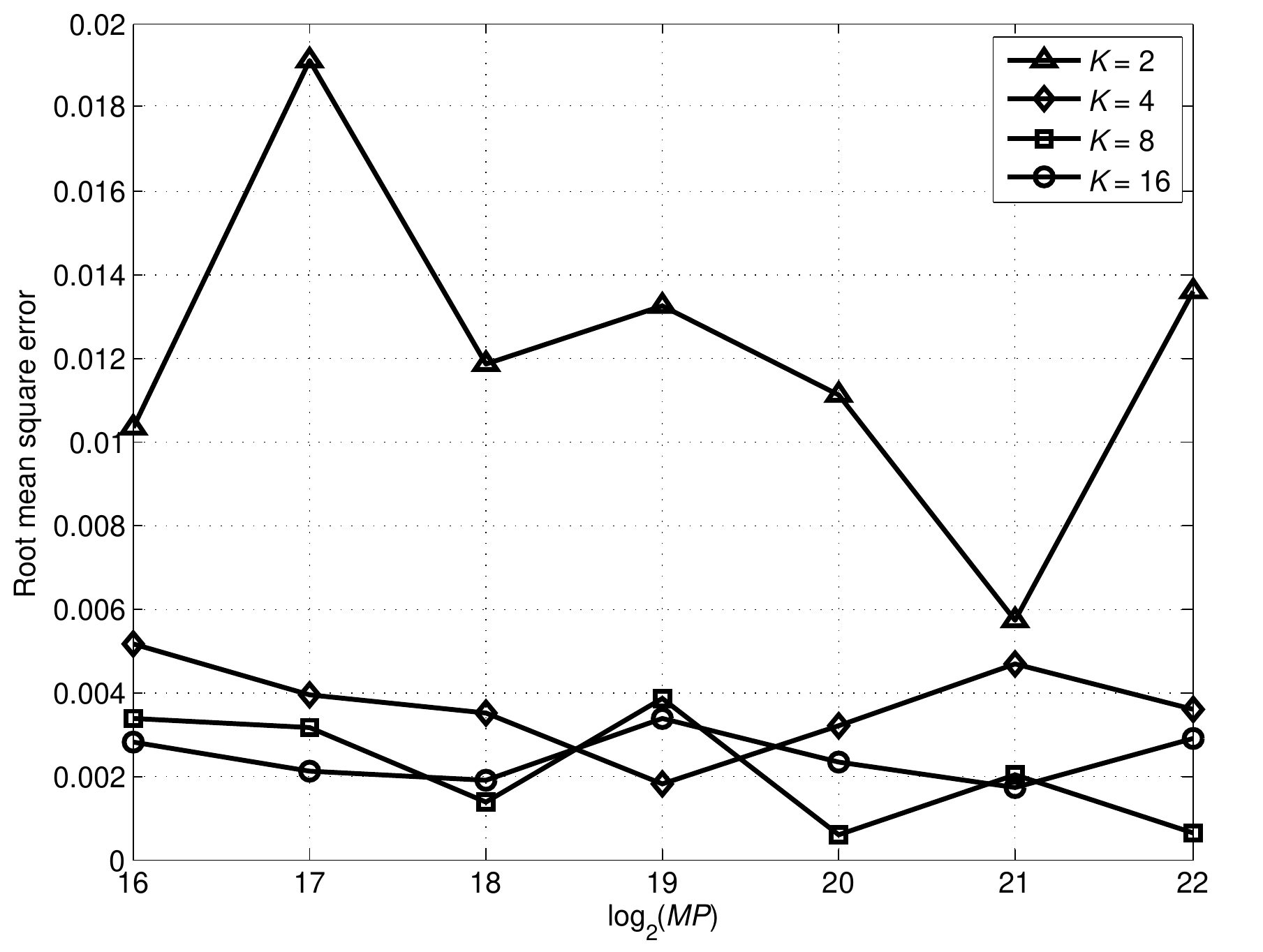}
\caption{SIFFT-based algorithm for exactly sparse multipath channel.}
\end{figure}

In Fig. 5, we give an example of $6$ different channels with deterministic nonzero coordinates and for each channel, the nonzero channel coefficients are i.i.d. complex Gaussian distribution. Suppose $L=256$, $M=8$, $D=32$, and the nonzero coordinates for Channel A: $\mathcal J_{\mathrm A}=\{0\}$, Channel B: $\mathcal J_{\mathrm B}=\{0,3\}$, Channel C: $\mathcal J_{\mathrm C}=\{0,3,8\}$, Channel D: $\mathcal J_{\mathrm D}=\{0,3,9\}$, Channel E: $\mathcal J_{\mathrm E}=\{0,3,9,19\}$, Channel F: $\mathcal J_{\mathrm F}=\{0,3,9,22\}$. Accordingly, the reminders of the nonzero coordinates modulo $M$ for Channel A: $\mathcal I_{\mathrm A}=\{0\}$, Channel B: $\mathcal I_{\mathrm B}=\{0,3\}$, Channel C: $\mathcal I_{\mathrm C}=\{0,3\}$, Channel D: $\mathcal I_{\mathrm D}=\{0,1,3\}$, Channel E: $\mathcal I_{\mathrm E}=\{0,1,3\}$, Channel F: $\mathcal I_{\mathrm F}=\{0,1,3,6\}$. It can be seen from Fig. 5 that the diversity order of Channel A is $1$, Channel B and Channel C are $2$, Channel D and Channel E are $3$, Channel F is $4$. It is pointed out that although $\left|\mathcal J_{\mathrm C}\right|=\left|\mathcal J_{\mathrm D}\right|$ and $\left|\mathcal J_{\mathrm E}\right|=\left|\mathcal J_{\mathrm F}\right|$, their corresponding diversity orders are different. As a result, we can verify the previous analysis that the diversity order of sparse multipath channel is determined by the cardinality of the set of reminder coordinates after mod $M$, rather than the cardinality of coordinate set itself.
\begin{figure}[t]
\centering
\includegraphics[width=3.5in]{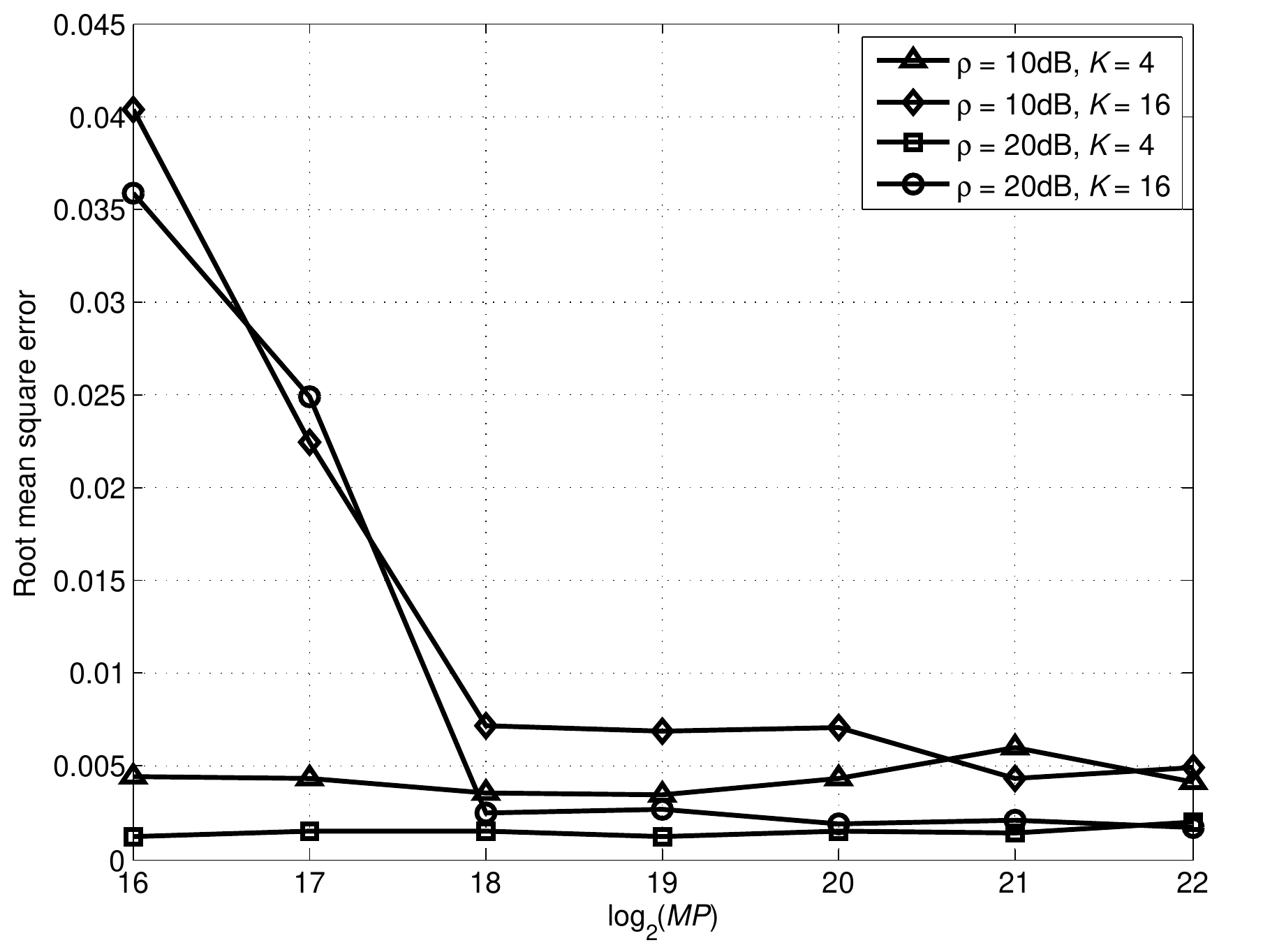}
\caption{SIFFT-based algorithm for approximately sparse multipath channel.}
\end{figure}

Figs. 6$-$8 show how the parameters $D$, $K$, $M$ influence on the BER performance of PIS decoding, respectively. Suppose the transmitted SNR $\rho=10\mathrm{dB}$. In Fig 6, we compare the BER with respect to the maximum delay $D$. It can be seen that the BER is roughly irrelevant to the variation of $D$, since $K$ nonzero taps are randomly distributed within $D+1$ taps. Fig. 7 investigates the relationship between the number of nonzero taps $K$ and the BER performance. Simulation result indicates that with the increase of $K$, the BER decreases almost linearly in the logarithmic scale. The reason is that the diversity order can be directly determined by $\kappa$ and increases with $K$ with large probability. Fig. 8 shows the BER performance with respect to VB size $M$ in the V-OFDM system. The BER first decreases with the increase of $M$ for $M\leqslant8$, whereas for $M\geqslant8$, the BER increases with $M$ instead. On one hand, a larger $M$ can avoid $K$ nonzero taps interacting with each other better after mod $M$, in this case, $\kappa=K$ with high probability which may improve the BER performance. On the other hand, for a given sphere radius $r$, with the increase of $M$, the probability that the transmitted symbols lie in the certain sphere decreases exponentially according to (30), and thus diminishes the advantage of the PIS decoding. Therefore, one can improve the BER performance of the PIS decoding by choosing an appropriate VB size $M$.

Fig. 9 compares different decoding approaches in the V-OFDM system. Suppose the parameters $D=16$, $K=4$, $L=256$, $M=4$, $r=\sqrt{\frac{\kappa}{\rho}\ln\rho}$. Since the condition $K\ll M$ does not hold, $\kappa<K$ with not a small probability which may diminish the multipath diversity orders of the MMSE decoding, ML decoding and PIS decoding. In fact, when $D$ is sufficiently large such that the reminders of the nonzero coordinates modulo $M$ can be regarded randomly distributed at the coordinates $0,1,\ldots,M-1$, for the given $K$ and $M$, the probability mass function of $\kappa$ is $P_{\kappa}=\binom{M}{\kappa}\binom{K}{\kappa}\big(\frac{\kappa}{M}\big)^K\kappa^{-\kappa}\kappa!,~\kappa=1,2,\ldots,\min\left\{K,M\right\}$. After averaging over the random nonzero coordinates of channel, the diversity orders of the ZF decoding, MMSE decoding, ML decoding and PIS decoding are corresponding to the minimum of $\kappa$ and thus equal to $1$. For the different decoding approaches, denote the BERs of ZF decoding, MMSE decoding, ML decoding, PIS decoding as $P_{\mathrm{ZF}}$, $P_{\mathrm{MMSE}}$, $P_{\mathrm{ML}}$, $P_{\mathrm{PIS}}$, respectively. It is well known that $P_{\mathrm{ML}}<P_{\mathrm{MMSE}}<P_{\mathrm{ZF}}$. Simulation result indicates that the PIS decoding loses certain BER performance since (36) does not hold if $\rho$ is not large enough, while for $\rho>10\mathrm{dB}$, the proposed PIS decoding outperforms the ZF decoding and MMSE decoding and gradually approximates to the ML decoding with the increase of $\rho$. Furthermore, the complexity of the PIS decoding decreases with $\rho$ and is much less than the ML decoding.
\begin{figure}[t]
\centering
\includegraphics[width=3.5in]{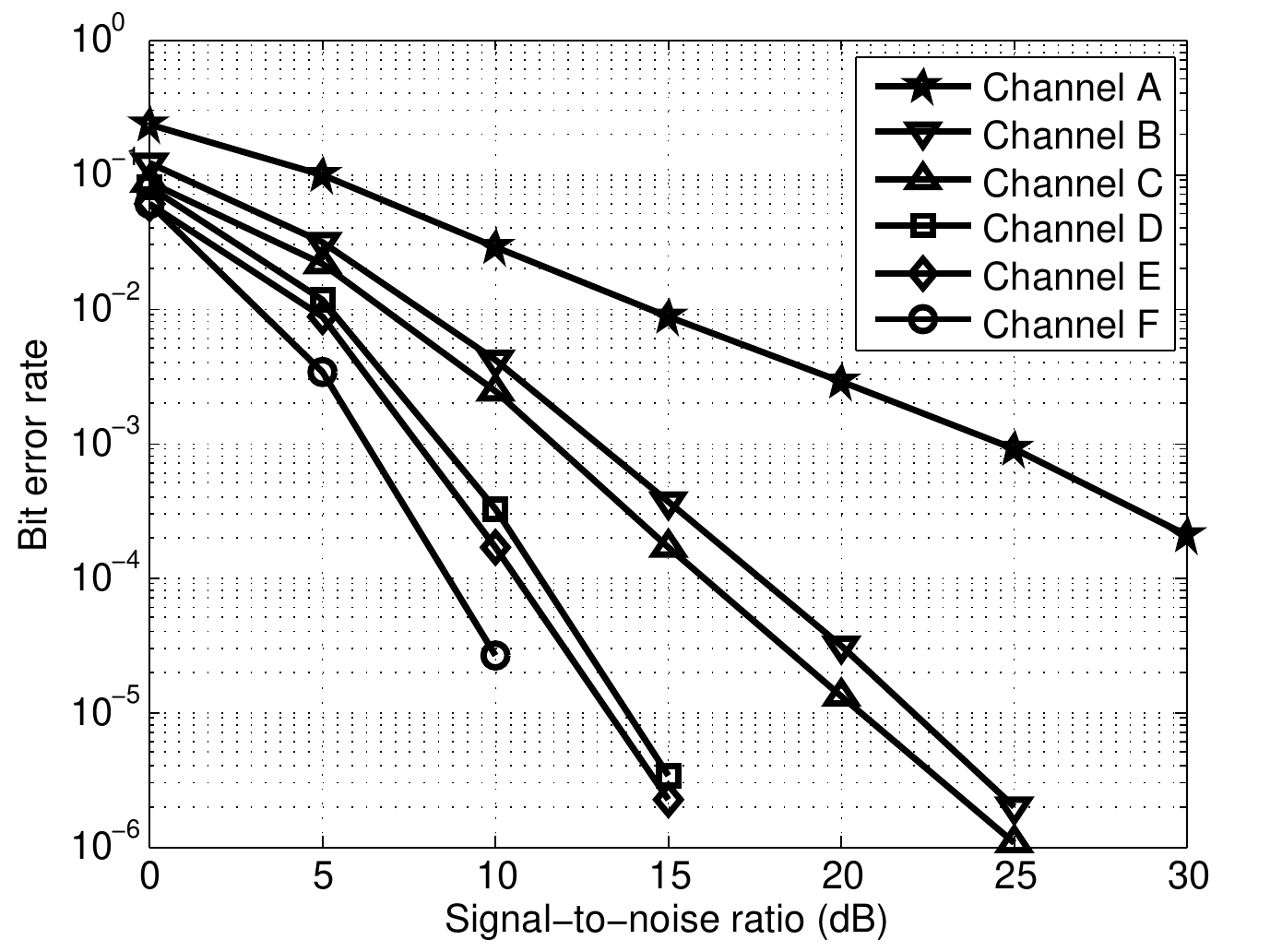}
\caption{Diversity orders for different channels with $L=256$ and $M=8$.}
\end{figure}
\begin{figure}[t]
\centering
\includegraphics[width=3.5in]{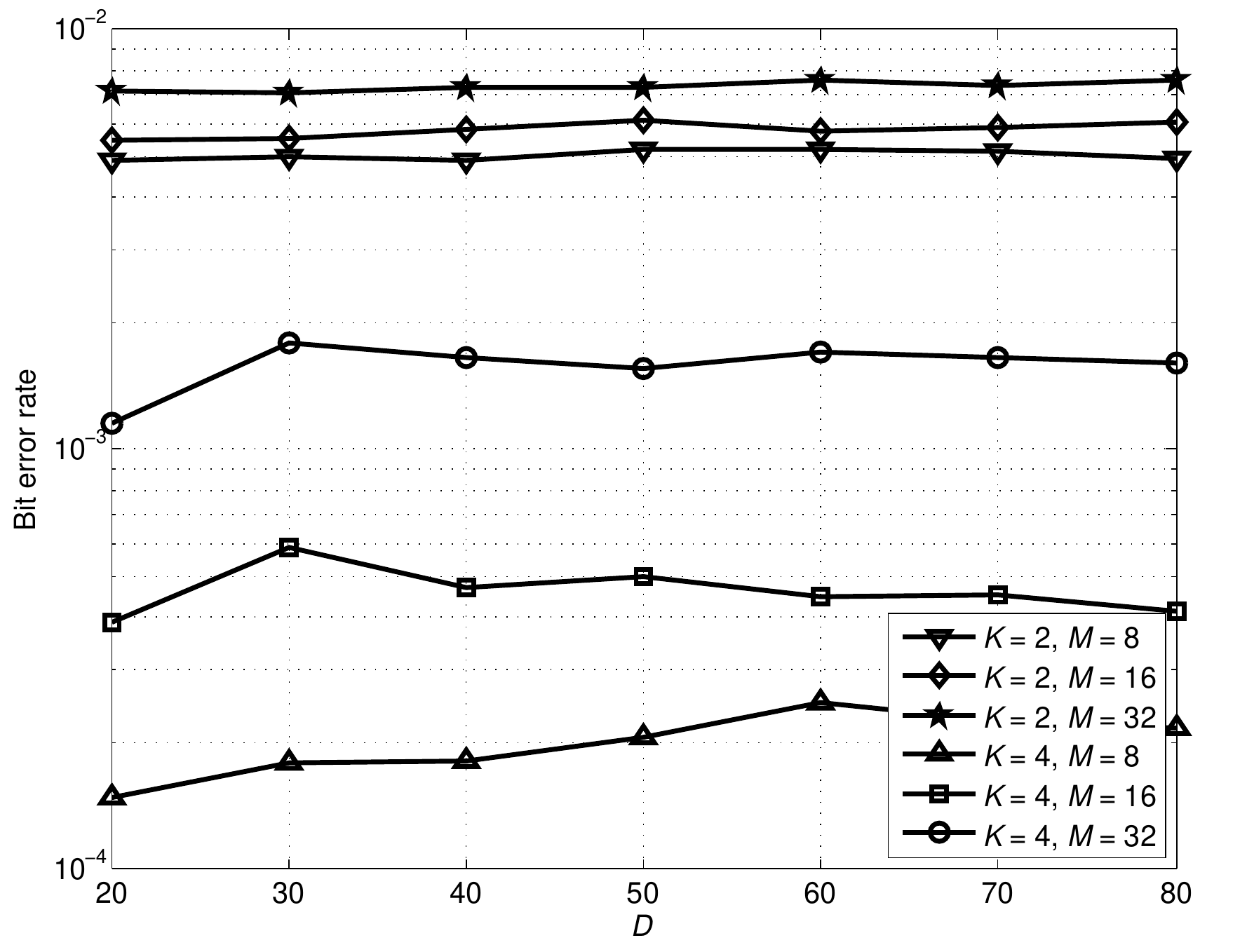}
\caption{PIS decoding for different $D$ with $\rho=10\mathrm{dB}$ and $r=\sqrt{\frac{\kappa}{\rho}\ln\rho}$.}
\end{figure}
\begin{figure}[t]
\centering
\includegraphics[width=3.5in]{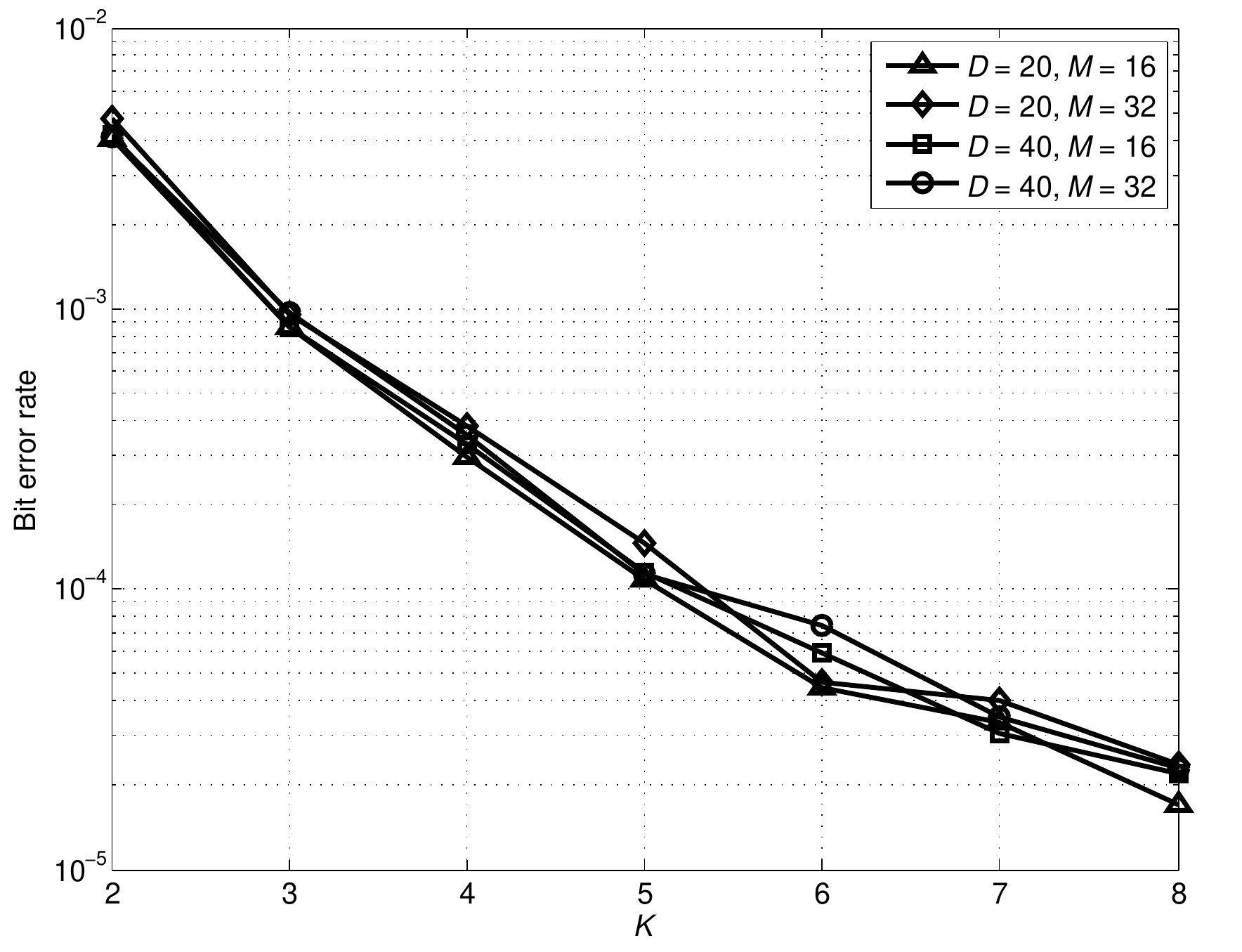}
\caption{PIS decoding for different $K$ with $\rho=10\mathrm{dB}$ and $r=\sqrt{\frac{\kappa}{\rho}\ln\rho}$.}
\end{figure}
\begin{figure}[t]
\centering
\includegraphics[width=3.5in]{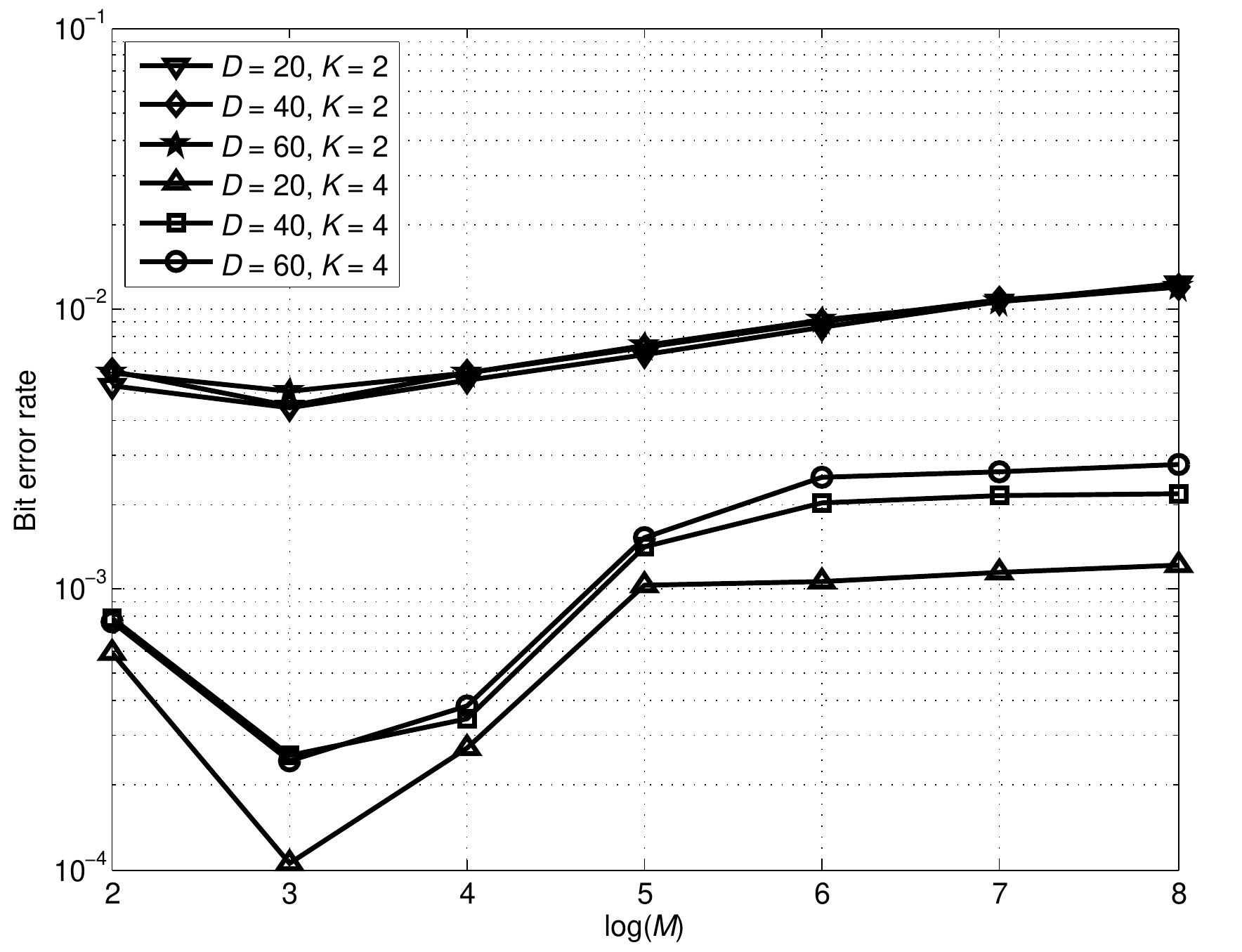}
\caption{PIS decoding for different $M$ with $\rho=10\mathrm{dB}$ and $r=\sqrt{\frac{\kappa}{\rho}\ln\rho}$.}
\end{figure}

In Fig. 10, we consider the channel estimation and decoding algorithms jointly. The BER performance is not only dependent of the decoding approaches, but also influenced by the channel estimation accuracy. Suppose the parameters $D=64$, $K=4$. For the OFDM system, the receiver employs the FFT-based interpolation for channel estimation and symbol-by-symbol decoding with parameters $N=1048576$, pilot channel number $P=65536$. For the V-OFDM system with linear receivers, we estimate channel by the conventional IFFT-based approach, and employ the ZF decoding and the MMSE decoding with parameters $L=131072$, $M=8$, $P=8192$, respectively. For the V-OFDM system with the ML decoding, we estimate the channel by the conventional IFFT-based approach as well with the parameters $L=262144$, $M=4$, $P=16384$. For the V-OFDM system with the PIS decoding, the SIFFT-based algorithm is employed for sparse multipath channel estimation with parameters $L=131072$, $M=8$, $P=8192$. If a slight bias is induced during the process of channel estimation, the sphere radius should not be extremely small since a robust sphere radius is needed to guarantee that the probability of the transmitted symbols lying in the sphere does not decrease. An empirical method to balance the tradeoff between the estimation error and the complexity is to choose sphere radius $r=\max\left\{\sqrt{\frac{\kappa}{\rho}\ln\rho},\sqrt{\frac{\kappa}{\rho_0}\ln\rho_0}\right\},~\rho_0=20\mathrm{dB}$. It can be seen from Fig. 10 that V-OFDM system outperforms the conventional OFDM system. Compared with the ZF decoding and the MMSE decoding, the proposed SIFFT-based channel estimation and PIS decoding reduces the BER significantly.
\begin{figure}[t]
\centering
\includegraphics[width=3.5in]{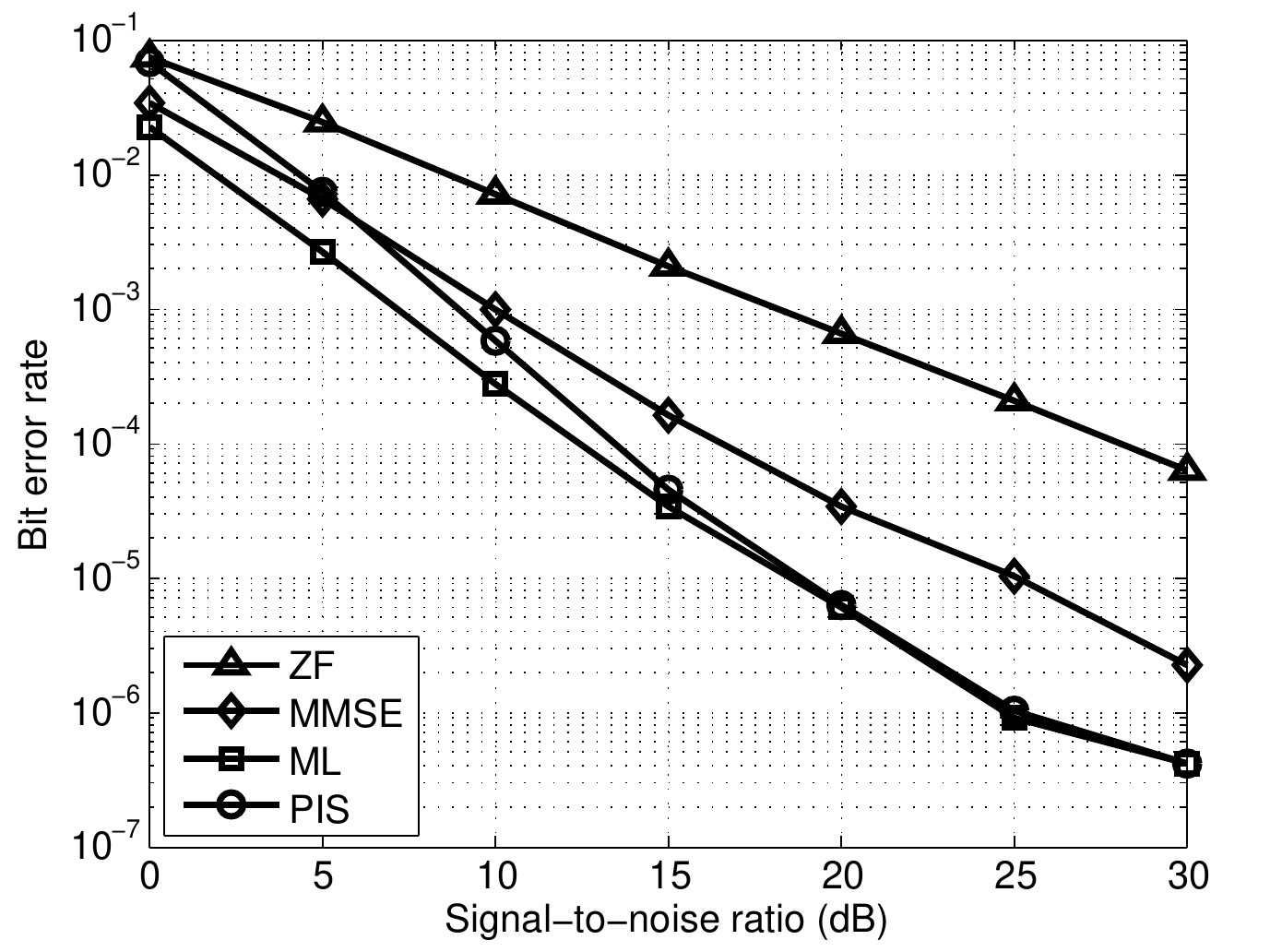}
\caption{Comparison of different decoding schemes with $L=256$ and $M=4$.}
\end{figure}

\section{Conclusion}
In this paper, we investigate sparse multipath channel estimation and decoding for broadband V-OFDM systems. From the system model, if the pilot channels are evenly allocated over multiple subcarriers, the pilot symbols are evenly distributed over the equivalent channels. For the sparse multipath channel estimation, we first design a type of pilot symbols that can minimize the MSE of an estimator. Then, we give SIFFT-based algorithms for exactly and approximately sparse multipath channel estimations corresponding to the cases with and without AWGN induced during the transmission, respectively. The remarkable significance of the SIFFT-based approach is to estimate the nonzero channel coefficients and their corresponding coordinates directly. For the PIS decoding algorithm, the diversity order is determined by not only the number of nonzero taps, but also the coordinates of nonzero taps. Simulation results indicate that the BER performance of the PIS decoding is comparable to that of the ML decoding with certain sphere radius for a sufficiently large SNR, but reduces the complexity substantially.


%

\appendices
\section{Proof of Sparse Multipath Channel With AWGN}
\begin{proof}
For the pilot signals being transmitted through multipath channel with AWGN, it was derived in (18) that $\widehat{\bm h}$ is an unbiased estimator of $\bm h$. The diagonal entries of (19) corresponding to the variance of $\widehat{\bm h}$, which are all equal to $\frac{\mathrm{tr}\left\{\mathbf\Sigma\right\}}{M^2P^2}$. Hence, $\widehat{\bm h}$ is a random vector and can be regarded as $\bm h$ with an additive noise whose entries are identically distributed with complex Gaussian noise, i.e., $\mathcal{CN}\big(0,\frac{\mathrm{tr}\left\{\mathbf\Sigma\right\}}{M^2P^2}\big)$, but may not be white.
For the sparse channel with only $K$ nonzero taps, $\eta$ approximates the expectation power ratio of $K$ dominant entries in $\widehat{\bm h}$ to the rest entries such that
\begin{equation}
\eta=\frac{\|\bm h\|_2^2+\frac{K}{M^2P^2}\mathrm{tr}\left\{\mathbf\Sigma\right\}}{\frac{MP-K}{M^2P^2}\mathrm{tr}\left\{\mathbf\Sigma\right\}}
\end{equation}
Considering the sparse channel that $K\ll MP$, (37) can thus be further simplified as
\begin{equation}
\eta=\frac{MP\|\bm h\|_2^2}{\mathrm{tr}\left\{\mathbf\Sigma\right\}}
\end{equation}

It was listed in Table I that the designed pilot symbols can reduce the interference efficiently. Suppose $\ell^2$ norm of the sparse channel is normalized, i.e., $\|\bm h\|_2=1$. Since $\mathrm{tr}\left\{\mathbf\Sigma\right\}$ is proportional to $\sigma^2$, we have
\begin{equation}
\eta=\frac{\rho}{\zeta}
\end{equation}
where the typical value of $\zeta$ is $1\sim2$ and roughly independent of $M$.

Therefore, it can be naturally concluded that if the designed pilot symbols are transmitted through the normalized sparse multipath channel with AWGN, then the estimator $\widehat{\bm h}$ is an approximately sparse vector with the parameter $\eta=\frac{\rho}{\zeta}$, where $\rho$ denotes the transmitted SNR, the typical value of $\zeta$ is $1\thicksim2$ and roughly independent of $M$.
\end{proof}

\begin{figure}[t]
\centering
\includegraphics[width=3.5in]{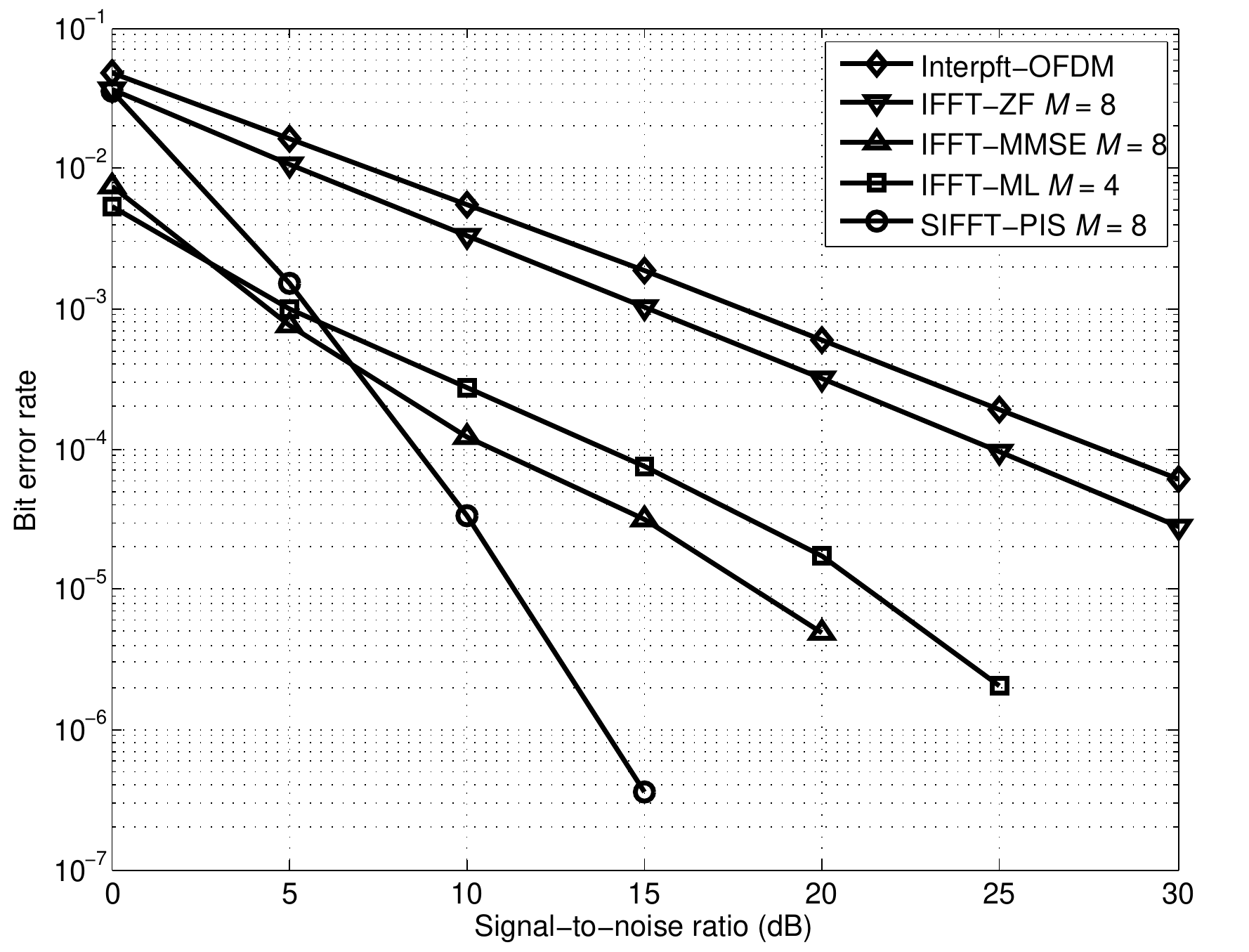}
\caption{Comparison of joint channel estimation and decoding algorithms.}
\end{figure}

\section{Proof of Diversity Order of ML Decoding for a Sparse Multipath Channel}
\begin{proof}
Assume the ideal channel state information is known at the receiver. Denote $\widehat{\bm X}_l$ as the estimation of $\bm X_l$ with the ML decoding, then the SER of the ML decoding conditioned on $\bm{\mathcal H}_l$ is written as $\Pr\big\{\widehat{\bm X}_l\neq\bm X_l\big|\bm{\mathcal H}_l\big\}$. According to the ML decoding, for the transmitted symbol $\bm X_l$ and a distinct symbol $\bm X_l'$ from the symbol constellation, if $\big\|\bm Y_l-\bm{\mathcal{H}}_l\bm X_l\big\|_2>\big\|\bm Y_l-\bm{\mathcal{H}}_l\bm X_l'\big\|_2$ holds, then the symbol error occurs. Since the event $\widehat{\bm X}_l\neq\bm X_l$ is equivalent to $\bigcup\limits_{\bm X_l'\in\mathbb X^M,\bm X_l'\neq\bm X_l}\big\|\bm Y_l-\bm{\mathcal{H}}_l\bm X_l\big\|_2>\big\|\bm Y_l-\bm{\mathcal{H}}_l\bm X_l'\big\|_2$ \cite{Zhang2012}, it is not difficult to derive a lower bound of the SER as
\begin{align}
&\Pr\big\{\widehat{\bm X}_l\neq\bm X_l\big|\bm{\mathcal H}_l\big\}\nonumber\\
&\geqslant\max\limits_{\bm X_l'\in\mathbb X^M,\bm X_l'\neq\bm X_l}\Pr\Big\{\big\|\bm Y_l-\bm{\mathcal{H}}_l\bm X_l\big\|_2>\big\|\bm Y_l-\bm{\mathcal{H}}_l\bm X_l'\big\|_2\Big|\bm{\mathcal H}_l\Big\}
\end{align}
and an upper bound of the SER as
\begin{align}
&\Pr\big\{\widehat{\bm X}_l\neq\bm X_l\big|\bm{\mathcal H}_l\big\}\nonumber\\
&\leqslant\sum\limits_{\bm X_l'\in\mathbb X^M,\bm X_l'\neq\bm X_l}\Pr\Big\{\big\|\bm Y_l-\bm{\mathcal{H}}_l\bm X_l\big\|_2>\big\|\bm Y_l-\bm{\mathcal{H}}_l\bm X_l'\big\|_2\Big|\bm{\mathcal H}_l\Big\}
\end{align}
Furthermore, the total number of elements in the symbol constellation $\mathbb X^M$ is $2^{RM}$, then the upper bound (41) can be further simplified as
\begin{align}
&\Pr\big\{\widehat{\bm X}_l\neq\bm X_l\big|\bm{\mathcal H}_l\big\}\leqslant\big(2^{RM}-1\big)\times\nonumber\\
&\max\limits_{\bm X_l'\in\mathbb X^M,\bm X_l'\neq\bm X_l}\Pr\Big\{\big\|\bm Y_l-\bm{\mathcal{H}}_l\bm X_l\big\|_2>\big\|\bm Y_l-\bm{\mathcal{H}}_l\bm X_l'\big\|_2\Big|\bm{\mathcal H}_l\Big\}
\end{align}

Therefore, the upper and low bounds of the SER have the same tendency and only differ by a constant multiplier. Denote $\bm e_l=\bm X_l'-\bm X_l$, then we have
\begin{equation}
\Pr\Big\{\big\|\bm Y_l-\bm{\mathcal{H}}_l\bm X_l\big\|_2>\big\|\bm Y_l-\bm{\mathcal{H}}_l\bm X_l'\big\|_2\Big|\bm{\mathcal H}_l\Big\}=Q\Big(\frac{\left\|\bm{\mathcal H}_l\bm e_l\right\|_2^2}{2\sigma^2}\Big)
\end{equation}
where the $Q$-function is define as $Q(x)=\frac{1}{\sqrt{2\pi}}\int_x^{+\infty}\mathrm e^{-\frac{t^2}{2}}\mathrm dt$. Substituting (8) into (43) and note that the $\ell^2$ distance does not change after the unitary transformation, (43) can be simplified as
\begin{equation}
\Pr\Big\{\big\|\bm Y_l-\bm{\mathcal{H}}_l\bm X_l\big\|_2>\big\|\bm Y_l-\bm{\mathcal{H}}_l\bm X_l'\big\|_2\Big|\bm{\mathcal H}_l\Big\}=Q\Big(\frac{\left\|\bm H_l\bm U_l\bm e_l\right\|_2^2}{2\sigma^2}\Big)
\end{equation}

Consider the sparse channel $\bm h$ has only $K$ i.i.d. nonzero taps and each nonzero entry is a complex Gaussian random variable with zero mean and unit variance. Recall that $\mathcal J$ is the set of coordinates of the nonzero taps. Suppose $j_0,j_1,\ldots,j_{K-1}$ are the $K$ entries in $\mathcal J$ with the ascending order $0\leqslant j_0<j_1<\ldots<j_{K-1}\leqslant D$. We construct a $K\times1$ vector $\widetilde{\bm h}=\big[h_{j_0},h_{j_1},\ldots,h_{j_{K-1}}\big]^{\mathrm T}$ and all the $K$ entries in $\widetilde{\bm h}$ are i.i.d. Gaussian variables. Denote $\widetilde{\bf F}_l$ as an $M\times K$ matrix constructed by extracting the $l$th, $(l+L)$th, $\ldots$ , $[l+(M-1)L]$th rows, $j_0$th, $j_1$th, $\ldots$ , $j_{K-1}$th columns of the $N$-point FFT matrix without normalization, i.e., $\big[\widetilde{\bf F}_l\big]_{r,c}=\mathrm e^{-\mathrm j\frac{2\mathrm\pi}{N}(l+rL)j_c}$. Then, we have $\bm H_l=\mathrm{diag}\big\{\widetilde{\bf F}_l\widetilde{\bm h}\big\}$. Therefore, (44) can be further rewritten as
\begin{equation}
\Pr\Big\{\big\|\bm Y_l-\bm{\mathcal{H}}_l\bm X_l\big\|_2>\big\|\bm Y_l-\bm{\mathcal{H}}_l\bm X_l'\big\|_2\Big|\bm{\mathcal H}_l\Big\}=Q\Big(\frac{\widetilde{\bm h}^{\mathrm H}{\bf E}_l^{\mathrm H}{\bf E}_l\widetilde{\bm h}}{2\sigma^2}\Big)
\end{equation}
where ${\bf E}_l=\mathrm{diag}\left\{{\bf U}_l\bm e_l\right\}\widetilde{\bf F}_l$. Denote $r_l$ as the rank of ${\bf E}_l^{\mathrm H}{\bf E}_l$, i.e., $r_l=\mathrm{rank}\left({\bf E}_l^{\mathrm H}{\bf E}_l\right)$, and $\lambda_{0,l},\lambda_{1,l},\ldots,\lambda_{r_l-1,l}$ are the $r_l$ nonzero eigenvalues corresponding to such positive semidefinite Hermitan matrices ${\bf E}_l^{\mathrm H}{\bf E}_l$ with the descending order $\lambda_{0,l}\geqslant\lambda_{1,l}\geqslant\ldots\geqslant\lambda_{r_l-1,l}>0$. Denote the constellation of pairwise error $\mathcal E=\big\{\bm e_l\big|\bm e_l=\bm X_l'-\bm X_l\neq\bm 0,\bm X_l\in\mathbb X^M,\bm X_l'\in\mathbb X^M\big\}$. After averaging over the complex Gaussian random channel $\widetilde{\bm h}$, an upper bound of the SER (42) can be calculated by the Chernoff bound as \cite{Han2010,Tarokh1998,Su2004}
\begin{align}
\Pr\big\{\widehat{\bm X}_l\neq\bm X_l\big\}&\leqslant\big(2^{RM}-1\big)\max\limits_{\bm e_l\in\mathcal E}\mathbb E\Big\{\mathrm e^{-\frac{\widetilde{\bm h}^{\mathrm H}{\bf E}_l^{\mathrm H}{\bf E}_l\widetilde{\bm h}}{4\sigma^2}}\Big\}\nonumber\\
&\leqslant\big(2^{RM}-1\big)\Big[\Big(\prod\limits_{i=0}^{r_{l\min}-1}\lambda_{i,l}\Big)^{\frac{1}{r_{l\min}}}\frac{1}{4\sigma^2}\Big]^{-r_{l\min}}
\end{align}
where $r_{l\min}\triangleq\min\limits_{\bm e_l\in\mathcal E}r_l$. Denote $r_{\min}\triangleq\min\limits_{l\in\mathcal L_{\mathrm{D}}}r_{l\min}$ and $\mathcal L_{\mathrm{D}}\triangleq\big\{l\big|0\leqslant l\leqslant L-1,~l\neq0,\frac{L}{P},\ldots,\frac{P-1}{P}L\big\}$. Hence, the SNR of the ML decoding is exponentially less than or equal to the minimum rank of ${\bf E}_l^{\mathrm H}{\bf E}_l$, i.e., $P_{\mathrm{ML}}~\dot\leqslant~\rho^{-r_{\min}}$.

To derive a lower bound of the SER, (40) can be further calculated by the recursion of the integration by parts
\begin{align}
\Pr\big\{\widehat{\bm X}_l\neq\bm X_l\big\}&\geqslant\max\limits_{\bm e_l\in\mathcal E}\mathbb E\Big\{Q\Big(\frac{\widetilde{\bm h}^{\mathrm H}{\bf E}_l^{\mathrm H}{\bf E}_l\widetilde{\bm h}}{2\sigma^2}\Big)\Big\}\nonumber\\
&\geqslant\max\limits_{\bm e_l\in\mathcal E}\sqrt{\frac{\lambda_{0,l}}{\lambda_{0,l}+4\sigma^2}}\sum_{k=r_l}^{+\infty}\frac{(2k-1)!!}{2^{k+1}k!(1+\frac{\lambda_{0,l}}{4\sigma^2})^k}\nonumber\\
&=\sqrt{\frac{\lambda_{0,l}}{\lambda_{0,l}+4\sigma^2}}\sum_{k=r_{l\min}}^{+\infty}\frac{(2k-1)!!}{2^{k+1}k!(1+\frac{\lambda_{0,l}}{4\sigma^2})^k}
\end{align}
where the factorial $k!=1\times2\times\cdots\times k$, and the double factorial $(2k-1)!=1\times3\times\cdots\times(2k-1)$. Then, the SNR of the ML decoding is exponentially greater than or equal to the minimum rank of ${\bf E}_l^{\mathrm H}{\bf E}_l$, i.e., $P_{\mathrm{ML}}~\dot\geqslant~\rho^{-r_{\min}}$.

According to (46) and (47), we conclude that the diversity order of the ML decoding is equal to the minimum rank of ${\bf E}_l^{\mathrm H}{\bf E}_l$ among all non-all-zero $\bm e_l$ and all $l\in\mathcal L_{\mathrm{D}}$, which is also denoted by $r_{\min}=\min\limits_{l\in\mathcal L_{\mathrm{D}}}\min\limits_{\bm e_l\in\mathcal E}\mathrm{rank}\left({\bf E}_l^{\mathrm H}{\bf E}_l\right)$. Note that
\begin{equation}
\mathrm{rank}\left({\bf E}_l^{\mathrm H}{\bf E}_l\right)=\mathrm{rank}\left({\bf E}_l\right)=\mathrm{rank}\big(\mathrm{diag}\left\{{\bf U}_l\bm e_l\right\}\widetilde{\bf F}_l\big)
\end{equation}

Now, we consider two matrices $\mathrm{diag}\left\{{\bf U}_l\bm e_l\right\}$ and $\widetilde{\bf F}_l$ separately. Note that ${\bf U}_l={\bf F}_M{\bf\Lambda}_l$ and ${\bf\Lambda}_l$ is a rotation matrix evolved by V-OFDM modulation itself. According to Theorem 2 in \cite{Cheng2011}, for the pulse-amplitude modulation (PAM) or BPSK modulation, $\min\limits_{\bm e_l\neq0}\mathrm{rank}\left(\mathrm{diag}\left\{{\bf U}_l\bm e_l\right\}\right)=M$ for $l=1,2,\ldots,L-1$, while for the quadrature amplitude modulation (QAM) or quadrature phase-shift keying (QPSK) modulation, $\min\limits_{\bm e_l\neq0}\mathrm{rank}\left(\mathrm{diag}\left\{{\bf U}_l\bm e_l\right\}\right)=M$ for $l=1,2,\ldots,\frac{L}{2}-1,\frac{L}{2}+1,\ldots,L-1$. In other words, for the conventional modulation, $\mathrm{diag}\left\{{\bf U}_l\bm e_l\right\}$ has full rank in most subchannels. Furthermore, from the previous analysis in Section II B, the $P$ subchannels with the indices $0,\frac{L}{P},\ldots,L-\frac{L}{P}$ are allocated to transmit pilot symbols. Hence, $\mathrm{diag}\left\{{\bf U}_l\bm e_l\right\}$ always has full rank for the channels allocated to transmit data symbols.

On the other hand, $\widetilde{\bf F}_l$ can be constructed as $\big[\widetilde{\bm F}_0,\widetilde{\bm F}_1,\ldots,\widetilde{\bm F}_{K-1}\big]$, where the column vector $\widetilde{\bm F}_q=\big[\mathrm e^{-\mathrm j\frac{2\mathrm\pi}{N}lj_q},\mathrm e^{-\mathrm j\frac{2\mathrm\pi}{N}(l+L)j_q},\ldots,\mathrm e^{-\mathrm j\frac{2\mathrm\pi}{N}[l+(M-1)L]j_q}\big]^{\mathrm T}$. Recall that $\mathcal I$ is the coordinates of the nonzero taps modulo $M$ and $i_0,i_1,\ldots,i_{\kappa-1}$ are the $\kappa$ entries in $\mathcal I$ with the ascending order $0\leqslant i_0<i_1<\ldots<i_{\kappa-1}\leqslant M-1$. According to (27), $\forall q\in\{0,1,\ldots,K-1\}$, $\exists p\in\{0,1,\ldots,\kappa-1\}$ and an integer $k$ such that $j_q=i_p+kM$. Then, $\widetilde{\bm F}_q=\mathrm e^{-\mathrm j\frac{2\mathrm\pi}{L}k}\widehat{\bm F}_p$, where $\widehat{\bm F}_p=\big[\mathrm e^{-\mathrm j\frac{2\mathrm\pi}{N}li_p},\mathrm e^{-\mathrm j\frac{2\mathrm\pi}{N}(l+L)i_p},\ldots,\mathrm e^{-\mathrm j\frac{2\mathrm\pi}{N}[l+(M-1)L]i_p}\big]^{\mathrm T}$. If there exists a distinct $q'$ and an integer $k'$ such that $j_{q'}=i_p+k'M$ still holds, then $\widetilde{\bm F}_{q'}=\mathrm e^{-\mathrm j\frac{2\mathrm\pi}{L}k'}\widehat{\bm F}_p=\mathrm e^{-\mathrm j\frac{2\mathrm\pi}{L}(k'-k)}\widetilde{\bm F}_q$. Hence, the vectors $\widetilde{\bm F}_q$ and $\widetilde{\bm F}_{q'}$ are linearly dependent of $\widehat{\bm F}_p$. Since there are only $\kappa$ such $p$ in total, if let $\widehat{\bf F}_l=\big[\widehat{\bm F}_0,\widehat{\bm F}_1,\ldots,\widehat{\bm F}_{\kappa-1}\big]$, we have $\mathrm{rank}\big(\widetilde{\bf F}_l\big)=\mathrm{rank}\big(\widehat{\bf F}_l\big)$. Since the vectors in each column of a DFT matrix are linearly independent, the vector columns $\widehat{\bm F}_p,~\forall p\in\{0,1,\ldots,\kappa-1\}$ which are equal to the $i_p$th columns of the $M$-point FFT matrix without normalization while multiplied by the factor $\mathrm e^{-\mathrm j\frac{2\mathrm\pi}{N}li_p}$ are also linearly independent. Then, the maximal number of linearly independent columns of $\widehat{\bf F}_l$ is $\kappa$. Thus, we have $\mathrm{rank}\big(\widetilde{\bf F}_l\big)=\mathrm{rank}\big(\widehat{\bf F}_l\big)=\kappa$.

Since $\mathrm{diag}\left\{{\bf U}_l\bm e_l\right\}$ is an invertible matrix, $\mathrm{rank}\big(\mathrm{diag}\left\{{\bf U}_l\bm e_l\right\}\widetilde{\bf F}_l\big)=\mathrm{rank}\big(\widetilde{\bf F}_l\big)=\kappa$. As a result, it can be concluded that the diversity order of the ML decoding for a sparse multipath channel is $\kappa$.
\end{proof}

\section{Proof of $P_{\mathcal X^{(M)}}\leqslant P_{\mathrm{ML}}$}
\begin{proof}
Denote the estimation symbol sequence $\widehat{\bm X}_l'$ as the conventional ML decoding of $\bm X_l$. In contrast to the PIS decoding, the SER of the ML decoding can be written as
\begin{equation}
P_{\mathrm{ML}}=\Pr\big\{\widehat{\bm X}_l'\neq\bm X_l\big|\bm X_l\in\mathbb X^M\big\}
\end{equation}
Since $\mathcal X^{(M)}\subseteq\mathbb X^M$, we have
\begin{align}
&P_{\mathcal X^{(M)}}=\Pr\big\{\widehat{\bm X_l}\neq\bm X_l\big|\bm X_l\in\mathcal X^{(M)}\big\}\nonumber\\
&=1-\Pr\Big\{\bm X_l=\mathop{\arg\min}\limits_{\bm X^{(M)}\in\mathcal X^{(M)}}\big\|\bm Y_l-\bm{\mathcal{H}}_l\bm X^{(M)}\big\|_2\Big|\bm X_l\in\mathcal X^{(M)}\Big\}\nonumber\\
&\leqslant1-\Pr\Big\{\bm X_l=\mathop{\arg\min}\limits_{\bm X^{(M)}\in\mathbb X^M}\big\|\bm Y_l-\bm{\mathcal{H}}_l\bm X^{(M)}\big\|_2\Big|\bm X_l\in\mathcal X^{(M)}\Big\}\nonumber\\
&=1-\Pr\Big\{\bm X_l=\mathop{\arg\min}\limits_{\bm X^{(M)}\in\mathbb X^M}\big\|\bm Y_l-\bm{\mathcal{H}}_l\bm X^{(M)}\big\|_2\Big|\bm X_l\in\mathbb X^M\Big\}\nonumber\\
&=\Pr\big\{\widehat{\bm X}_l'\neq\bm X_l\big|\bm X_l\in\mathcal X^{(M)}\big\}\nonumber\\
&=P_{\mathrm{ML}}
\end{align}
\end{proof}
\section{Proof of Theorem 1}
\begin{proof}
Assume sparse channel $\bm h$ has only $K$ i.i.d. nonzero taps and each channel coefficient follows complex Gaussian random distribution. It is not difficult to find that the nonzero entries in $\bm{\mathcal H}_l$ are also complex Gaussian random variables but the variances may not be equal. Recall that $\mathcal S^{(m)}$ was defined in (28) that all possible symbol sequences lying in the certain sphere of radius $r$ around the received signal $\bm Y_l^{(m)}$. $\bm S_{\dagger}^{(m)}$ is the correct symbol sequence corresponding to the transmitted symbols. $d^{(m)}$ is the distance between $\bm Y_l^{(m)}$ and $\bm{\mathcal{H}}_l^{(m)}\bm S_{\dagger}^{(m)}$, i.e., $d^{(m)}=\big|\bm Y_l^{(m)}-\bm{\mathcal{H}}_l^{(m)}\bm S_{\dagger}^{(m)}\big|$. Since the noise $\bm{\Xi}_l$ in (6) is complex AWGN, $d^{(m)}$ is Rayleigh distributed with mean $\frac{\sqrt{\pi}}{2}\sigma$ and variance $\frac{4-\pi}{4}\sigma^2$. According to the triangle inequality, $\forall\bm S\in\mathbb X^{\kappa}$, we have
\begin{equation}
\big|\bm Y_l^{(m)}-\bm{\mathcal{H}}_l^{(m)}\bm S\big|+\big|\bm Y_l^{(m)}-\bm{\mathcal{H}}_l^{(m)}\bm S_{\dagger}^{(m)}\big|\geqslant\big|\bm{\mathcal{H}}_l^{(m)}\big(\bm S-\bm S_{\dagger}^{(m)}\big)\big|
\end{equation}

Let $\mathcal S_{\dagger}^{(m)}$ be the set of a type of symbol sequence $\bm S$ such that the distance between $\bm{\mathcal{H}}_l^{(m)}\bm S$ and $\bm{\mathcal{H}}_l^{(m)}\bm S_{\dagger}^{(m)}$ is less than or equal to the sphere of radius $r+d^{(m)}$, i.e.,
\begin{equation}
\mathcal S_{\dagger}^{(m)}=\left\{\bm S\Big|\big|\bm{\mathcal{H}}_l^{(m)}\big(\bm S-\bm S_{\dagger}^{(m)}\big)\big|\leqslant r+d^{(m)}\right\}
\end{equation}

Compared with $\mathcal S^{(m)}$ defined in (28), it can be found that $\mathcal S^{(m)}\subseteq\mathcal S_{\dagger}^{(m)}$. If $r$ is chosen exponentially equal to $\sqrt{\frac{\kappa}{\rho}\ln\rho}$ such that Lemma 1 is satisfied, then $r\rightarrow0$ as $\rho\rightarrow\infty$. Due to Rayleigh distribution, $d^{(m)}\rightarrow0$ with probability $1$ as $\rho\rightarrow\infty$. Since the probability density function of each nonzero entry in $\bm{\mathcal H}_l$ is a complex Gaussian function, for any given $\bm S\in\mathcal S_{\dagger}^{(m)}$, if $\bm S\neq\bm S_{\dagger}^{(m)}$, we have
\begin{align}
&\Pr\Big\{\big|\bm{\mathcal{H}}_l^{(m)}\big(\bm S-\bm S_{\dagger}^{(m)}\big)\big|\leqslant r+d^{(m)}\Big|\bm S\neq\bm S_{\dagger}^{(m)}\Big\}\nonumber\\
&=1-\mathrm e^{-\frac{\left(r+d^{(m)}\right)^2}{\left\|\bm S-\bm S_{\dagger}^{(m)}\right\|_2^2}}
\end{align}

It is found that for $\rho\rightarrow\infty$, $r+d^{(m)}\rightarrow0$ such that (53) approximates to $0$. Since $\mathcal S_{\dagger}^{(m)}$ has only a bounded finite elements, we have $\Pr\big\{\bm S\neq\bm S_{\dagger}^{(m)}\big|\bm S\in\mathcal S_{\dagger}^{(m)}\big\}=0$, thus, $\mathcal S_{\dagger}^{(m)}$ has only one entry $\bm S_{\dagger}^{(m)}$ with probability $1$ when $\rho\rightarrow\infty$. Obviously, $\bm S_{\dagger}^{(m)}\in\mathcal S^{(m)}$ with probability 1 as $\rho\rightarrow\infty$. Since $\mathcal S^{(m)}$ is a subset of $\mathcal S_{\dagger}^{(m)}$, $\mathcal S^{(m)}$ also has only one entry $\bm S_{\dagger}^{(m)}$ with probability $1$.

For any given small sphere radius $r$, Lemma 1 is satisfied as $\rho\rightarrow\infty$ such that the PIS decoding achieves the same multipath diversity order as the ML decoding, which is equal to $\kappa$. Furthermore, the cardinality $\big|\mathcal X^{(m)}\big|$ of the set $\mathcal X^{(m)}$ of possible symbol sequences in Algorithm 3 remains $1$ in each iteration that the complexity for updating $\mathcal X_{\bm S}^{(m)}$ can be reduced significantly. For $\big|\mathcal X^{(m)}\big|=1$, the evaluation and comparison operations are performed $\mathcal O(\kappa)$ times in the $m$th iteration. Consider $M$ iterations and $L$ subchannels, the complexities of the evaluation and comparison operations are $\mathcal O(\kappa LM)$ with probability $1$. In the previous analysis of the PIS decoding, the complexity of complex multiplication operation is $\mathcal O(\kappa LM2^{R\kappa})$. Since the evaluation and comparison operations for a complex number are faster than a complex multiplication operation, the total complexity of the PIS decoding is $\mathcal O(\kappa LM2^{R\kappa})$ with probability $1$.
\end{proof}


\ifCLASSOPTIONcaptionsoff
  \newpage
\fi



%

\balance

%








\end{document}